\documentclass[journal]{IEEEtran}
\usepackage{graphicx}
\usepackage{amsmath,amssymb,amsfonts,amsthm}
\usepackage[noend]{algpseudocode}
\usepackage{algorithmicx,algorithm}
\usepackage{subfigure}
\usepackage{cite}
\usepackage{setspace}
\usepackage{stfloats}
\usepackage{enumitem}
\usepackage{threeparttable}
\usepackage{comment}
\usepackage{bm}
\usepackage{booktabs}
\usepackage{multirow}
\usepackage{makecell}
\usepackage[table]{xcolor}
\usepackage{tabularx}

\newtheorem{Theorem}{Theorem}
\newtheorem{Definition}{Definition}
\newtheorem{Corollary}{Corollary}
\newtheorem{Proposition}{Proposition}

\newtheorem{remark}{Remark}
\graphicspath{{images/}}

\begin{document}	

\title{Fluid-Spatiotemporal Stochastic Geometry for Information Flow in Non-Stationary Fields: Kinematics and Field Reconstruction}

	\author{Wen-Yu Dong, Weiwei Jiang,~\IEEEmembership{Senior Member,~IEEE}, Song Zhao,\\ Qi Bi,~\IEEEmembership{Fellow,~IEEE}, Sheng Chen,~\IEEEmembership{Life Fellow,~IEEE}	%

		\thanks{W.-Y. Dong, S. Zhao and Q. Bi are with Future Technology Research Center, China Telecom Research Institute, Beijing 102209, China (E-mails: dongwy@chinatelecom.cn; zhaosong1@chinatelecom.cn; qibi@chinatelecom.cn)} %
		\thanks{W. Jiang is with the School of Information and Communication Engineering, Beijing University of Posts and Telecommunications, Beijing, 100876, China (Email: jww@bupt.edu.cn)}
		\thanks{S. Chen is with the School of Electronics and Computer Science, University of Southampton, Southampton SO17 1BJ, U.K., and also with Faculty of Information Science and Technology, Ocean University of China, Qingdao 266100, China (E-mail: sqc@ecs.soton.ac.uk).} %
		\vspace*{-5mm}
	}
	
	\maketitle 

\begin{abstract}
	Classical stochastic geometry primarily addresses forward analysis by prescribing spatial point-process laws and mobility models. It does not directly address the inverse problem of recovering directional macroscopic transport from aggregate density evolution in non-stationary spatial networks. This paper develops the kinematic foundation of fluid-spatiotemporal stochastic geometry (F-STSG), representing an evolving mobile population by an intensity field under an assumed continuum limit while retaining the network infrastructure as a discrete spatial measure. We formulate macroscopic transport reconstruction as an inverse boundary value problem. Under a strictly positive bounded intensity and compatible source and boundary data, an instantaneous weighted minimum-energy criterion selects a unique velocity representative from the admissible class. The reconstruction reduces to a weighted elliptic Neumann problem and separates density-identifiable transport from weighted-solenoidal circulation that is not identifiable from density evolution alone. We further derive the information flux, interface transport, velocity-divergence relation, and centroid dynamics. Numerical experiments examine reconstruction consistency and sensitivity, while the public New Era of Urban Traffic Monitoring with Aerial Footage (pNEUMA) dataset provides held-out validation of internal transport. A separate unmanned aerial vehicle (UAV) study evaluates the centroid as a causal downstream control reference under observation delay, motion constraints, and execution perturbations.
\end{abstract}
	
	
\section{Introduction}\label{sec:intro} 

\IEEEPARstart{S}{tochastic} geometry (SG) has provided a rigorous framework for relating random spatial configurations of wireless nodes to coverage, interference, and scaling behavior \cite{1995Stochastic,Baccelli2009,Gupta2000,Haenggi2012}. In its classical form, SG represents the network as a spatial snapshot governed by a prescribed point-process law. Mobility-aware extensions introduce temporal evolution by specifying displacement or mobility models between successive spatial configurations \cite{Haenggi2013,Krishnan2017,Tabassum2019}. Both formulations primarily support forward analysis: once the spatial law and mobility mechanism are specified, the resulting wireless performance can be evaluated. In operational networks, however, individual trajectories and exact mobility laws are often unavailable. Network observations more commonly take the form of time-indexed occupancy, traffic, or load fields aggregated over spatial regions.

In non-stationary spatial networks, aggregate demand may evolve coherently across space rather than through independently resampled configurations. We refer to this structured macroscopic evolution as the ``Digital Tide,'' represented mathematically by a time-varying macroscopic intensity field $\lambda$. It encompasses the propagation, accumulation, depletion, and boundary exchange of spatial demand, all of which can influence handover, interference, and infrastructure load. An instantaneous density field indicates where demand is concentrated, while a sequence of such fields reveals how the spatial distribution changes over time. What remains unobserved is the directional transport that carries demand between locations. Recovering this transport from aggregate density evolution is therefore necessary to move beyond a sequence of spatial snapshots toward a macroscopic description of network dynamics.

We represent this macroscopic evolution by a time-varying intensity field $\lambda(\bm x,t)$ \cite{Lasry2007}. Its evolution obeys a conservation law relating the density to the information flux $\bm J=\lambda\bm v$, together with an internal source term and boundary exchange. Given the density evolution and source term, the conservation law determines the divergence of $\bm J$, while the prescribed boundary exchange constrains its normal component. These constraints do not, in general, uniquely determine the full vector field because divergence-free components with compatible boundary behavior may remain invisible to the observed density evolution. The central theoretical problem is therefore to select a physically meaningful macroscopic transport representative from these aggregate observations.

The present paper develops the kinematic foundation of fluid-spatiotemporal stochastic geometry (F-STSG). A broader preliminary formulation of F-STSG is available in the arXiv preprint \cite{Dong2026FSTSG}, whereas this paper focuses on its density-to-transport inverse problem. Specifically, we formulate a weighted inverse boundary value problem \cite{Isakov2006}, define the admissible class of velocity fields consistent with the estimated density evolution, specified source term, and prescribed normal boundary flux, and select its instantaneous minimum-energy element. Under explicit positivity, compatibility, boundary, and weighted-space assumptions, this selection yields a unique macroscopic velocity representative and separates the density-identifiable transport component from microscopic circulation that cannot be recovered from aggregate density observations.

Once the transport field is reconstructed, the selected velocity specifies a balance-consistent macroscopic transport representative, its information flux quantifies the corresponding signed transport across spatial regions, and its divergence characterizes compression and expansion within that representation. Together, these quantities provide a kinematic description of how spatial demand moves, accumulates, and crosses network boundaries, thereby forming a basis for downstream tasks such as handover analysis, load migration, routing, and resource placement. This paper focuses on reconstructing and characterizing these macroscopic kinematic quantities from aggregate field observations, rather than developing application-specific decision policies.
	
\subsection{Related Works}\label{sec:related_works} 

The inverse reconstruction studied in this paper is closely related to dynamic optimal transport, mobility-aware stochastic geometry, continuum and mean-field modeling, and spatiotemporal field estimation. These lines of work address different aspects of density evolution, transport, mobility, and field estimation.

\begin{table*}[b]
	\vspace*{-4mm}
	\scriptsize
	\centering
	\caption{Scope-based comparison of F-STSG with closely related technical approaches}
	\label{tab:gap_analysis} 
	\vspace*{-2mm}
	\renewcommand{\arraystretch}{1.16}
	\setlength{\tabcolsep}{2.4pt}
	\begin{threeparttable}
			\begin{tabular}{
					>{\raggedright\arraybackslash}p{6cm}
					>{\centering\arraybackslash}p{2.05cm}
					>{\centering\arraybackslash}p{2.05cm}
					>{\centering\arraybackslash}p{2.05cm}
					>{\centering\arraybackslash}p{2.05cm}
					>{\columncolor{gray!10}\centering\arraybackslash}p{2.05cm}}
				\toprule
				\multicolumn{1}{c}{\textbf{Technical Dimension}} &
				\textbf{\makecell{Mobility-Aware\\SG}} &
				\textbf{\makecell{Continuum Flow\\(MFG/LWR)}} &
				\textbf{\makecell{Graph and Data-\\Driven Models}} &
				\textbf{\makecell{Dynamic Optimal\\Transport}} &
				\textbf{\makecell{F-STSG\\(This Work)}} \\
				\midrule
				
				\textbf{1. Euclidean spatial state} &
				\checkmark & \checkmark & $\triangle$ & \checkmark & \textbf{\checkmark} \\
				\textit{Uses physical coordinates, distances, or spatial measures} &
				& & & & \\
				
				\textbf{2. Non-stationary evolution} &
				$\triangle$ & \checkmark & \checkmark & \checkmark & \textbf{\checkmark} \\
				\textit{Represents spatial states that evolve continuously over time} &
				& & & & \\
				
				\textbf{3. Density-to-transport inversion} &
				$\times$ & $\triangle$ & $\triangle$ & \checkmark & \textbf{\checkmark} \\
				\textit{Infers a transport field from observed density evolution} &
				& & & & \\
				
				\textbf{4. No prescribed microscopic law} &
				$\times$ & $\triangle$ & \checkmark & \checkmark & \textbf{\checkmark} \\
				\textit{Does not require individual trajectories or agent utilities} &
				& & & & \\
				
				\textbf{5. Conservation-constrained kinematics} &
				$\triangle$ & \checkmark & $\triangle$ & \checkmark & \textbf{\checkmark} \\
				\textit{Links density evolution to velocity, flux, and divergence} &
				& & & & \\
				
				\textbf{6. Field--infrastructure coupling} &
				$\triangle$ & $\triangle$ & $\triangle$ & $\times$ & \textbf{\checkmark} \\
				\textit{Couples a continuous demand field to discrete wireless resources} &
				& & & & \\

				\midrule
				\multicolumn{1}{c}{\textbf{Representative References}} &
				\cite{Baccelli2009,Haenggi2013} &
				\cite{Lasry2007,Kim2020,Lighthill1955} &
				\cite{Yu2018STGCN,Wu2019GraphWaveNet,Wang2023} &
				\cite{Benamou2000,Villani2003} &
				\textbf{This Work} \\
				\bottomrule
			\end{tabular}
			\begin{tablenotes}
				\footnotesize
				\item The entries summarize the typical formulations and primary objectives of the cited works rather than universal capabilities of every method in each broad class. \checkmark: explicitly represented as a characteristic component of the class; $\triangle$: present only in some representative formulations, under additional assumptions, or as a partial component; $\times$: outside the primary formulation of the cited representative works. MFG denotes mean-field games, LWR denotes the Lighthill--Whitham--Richards model, and F-STSG denotes fluid-spatiotemporal stochastic geometry.
		\end{tablenotes}
	\end{threeparttable}
	\vspace*{-1mm}
\end{table*}

\subsubsection{Dynamic optimal transport}
Dynamic optimal transport characterizes density evolution through a continuity constraint and a kinetic-energy objective over a density--velocity path \cite{Benamou2000,Villani2003}. Building on this variational structure, the present formulation conditions on an estimated wireless density evolution and reconstructs, at each observation time, the weighted minimum-energy element of the admissible velocity class under specified source and boundary data. The reconstructed continuum quantities are subsequently coupled to a discrete wireless infrastructure measure. The minimum-energy selection is performed independently at each observation time and does not require the observed density trajectory to be a finite-horizon Wasserstein geodesic.

\subsubsection{Mobility-aware stochastic geometry}
Following the foundational works in \cite{Baccelli2009,Haenggi2012}, SG has become a standard tool for spatial wireless analysis \cite{DongTCOM,DongJSAC,DongIoTJ,DongTCOM2,DongGC2}. Mobility-aware extensions introduce prescribed displacement, mobility, or time-varying intensity models to characterize handover, interference correlation, queueing dynamics, and temporal wireless performance \cite{Lu2021,Win2009,Yang2018,Zhong2018}. F-STSG considers the inverse direction: an estimated density evolution, together with source and boundary information, is used to reconstruct a density-identifiable macroscopic transport representative. Demand-aware unmanned aerial vehicle (UAV) deployment methods \cite{Zhang2022} are relevant to the subsequent control application, where they provide application-level references and baselines.

\subsubsection{Continuum and mean-field models}
Continuous-media models describe dense interacting systems through macroscopic states. Representative examples include vehicular-flow models \cite{Lighthill1955,Greenshields1935}, pedestrian-flow models \cite{Helbing1995,Hughes2002}, and mean-field game (MFG) formulations for wireless caching, computation offloading, and interference management \cite{Kim2020,Zheng2021,Zhang2019,Wang2014}. These models typically prescribe a velocity--density relation, cost functional, or agent objective and then solve a forward evolution, equilibrium, or control problem. The present formulation instead treats the density evolution, source term, and boundary exchange as given inputs and performs a fixed-time inverse selection without identifying individual trajectories or microscopic utilities.

\subsubsection{Spatiotemporal estimation and forecasting}
Temporal point processes, complex-network models, and data-driven methods support event modeling, traffic estimation, forecasting, and resource decisions \cite{Gonzalez2016,Boccaletti2006,Wang2023,Ding2025,Li2024}. Representative spatiotemporal graph models, including the spatiotemporal graph convolutional network (STGCN) \cite{Yu2018STGCN} and Graph WaveNet \cite{Wu2019GraphWaveNet}, learn temporal dependencies over a discrete graph. Physical distance or spatial adjacency can be incorporated through graph construction in specific formulations, but Euclidean geometry is not necessarily retained as an explicit state variable across the broader graph and data-driven class. Their estimated or forecast density fields can serve as inputs to the conservation-constrained reconstruction, while the source and boundary terms must be obtained from corresponding measurements or separately specified models.

Table~\ref{tab:gap_analysis} summarizes how the representative formulations treat spatial state, temporal evolution, density-to-transport reconstruction, conservation-constrained kinematics, and wireless infrastructure coupling.
F-STSG integrates these elements through a wireless density-to-transport inverse formulation, a weighted admissible-space analysis, and the resulting velocity, information-flux, and divergence quantities under the stated positivity, compatibility, and boundary assumptions. The reconstructed quantities provide a kinematic interface to subsequent wireless analysis and application-specific control.

\subsection{Main Contributions}\label{sec:contributions}
This paper establishes the kinematic and inverse-reconstruction foundation of F-STSG for non-stationary spatial networks. Starting from the conservation of an evolving intensity field, it develops a principled route from aggregate spatial observations to macroscopic transport, network-level kinematic quantities, and constrained control references. The main contributions are as follows:
\begin{itemize}
	\item \textbf{Conservation-based kinematic foundation:} We formulate non-stationary network evolution in terms of the intensity field $\lambda(\bm x,t)$, the information flux $\bm J(\bm x,t)=\lambda(\bm x,t)\bm v(\bm x,t)$, internal source terms, and boundary exchange. This formulation connects the microscopic user constellation to a continuum transport description and makes explicit the inverse question at the core of F-STSG: which part of the macroscopic motion is identifiable from aggregate density evolution?
	\item \textbf{Well-posed inverse reconstruction and identifiability:} We cast density-to-transport reconstruction as a weighted inverse boundary value problem. In the corresponding weighted admissible space, an instantaneous kinetic-energy criterion selects a unique velocity field under the stated positivity, compatibility, and boundary assumptions. The resulting weighted elliptic problem provides both a constructive reconstruction method and a precise separation between density-identifiable transport and circulation that is invisible to the density evolution.
	\item \textbf{Analytical characterization under spatial heterogeneity:} For a positive background intensity with weak spatial variation, we derive a first-order perturbation model of the weighted potential equation, together with its effective forcing, Green's-function representation, and spectral decomposition. The modal form reveals how spatial heterogeneity modifies the reconstructed transport through the eigenstructure of the background operator and extends the analysis beyond radially symmetric fields.
	\item \textbf{Network-level kinematic quantities and empirical evaluation:} From the reconstructed transport field, we derive signed interface transport, fixed-cell population balance, velocity divergence, the material-derivative relation, and centroid dynamics. Controlled experiments examine finite-population convergence and reconstruction sensitivity, while the real New Era of Urban Traffic Monitoring with Aerial Footage (pNEUMA) mobility traces provide held-out validation of internal crossing rates. A separate UAV study evaluates the centroid as a causal control reference under observation delay, motion constraints, and execution-layer nonidealities.
\end{itemize}
The remainder of this paper is organized as follows. Section~\ref{sec:system_model} develops the continuum model from the underlying user constellation and formulates the conservation law. Section~\ref{sec:core_theory} establishes the inverse-reconstruction theory and its minimum-energy solution. Section~\ref{sec:linear_response} derives the weak-variation Green's-function and spectral representations. Section~\ref{sec:metrics_toolbox} develops the kinematic quantities used for network analysis. Section~\ref{sec:simulation} evaluates finite-population convergence, matched reconstruction, sensitivity to observation and model perturbations, real-trace transport reconstruction, the derived kinematic identities, and centroid-reference UAV control under planning and execution constraints. Finally, Section~\ref{sec:conclusion} concludes the paper.

\begin{table}[!b]
	\vspace*{-4mm}
	\scriptsize
	\caption{Principal notation used across sections}
	\label{tab:key_notations}
	\vspace*{-2mm}
	\centering
	\renewcommand{\arraystretch}{1.15}
	\setlength{\tabcolsep}{2pt}
	\begin{tabularx}{\columnwidth}{
			@{}
			>{\raggedright\arraybackslash}p{0.27\columnwidth}
			>{\raggedright\arraybackslash}X
			@{}
		}
		\toprule
		\textbf{Symbol} & \textbf{Definition} \\
		\midrule
		$\Omega,\partial\Omega,\bm n;\ \Phi_{\mathrm B}$ &
		Observation domain, boundary, outward normal; base-station process \\
		$\Phi_{\mathrm M}^{(n)}(t),\mu_n,p$ &
		Mobile-user process, empirical measure, and limiting density \\
		$\lambda,\widehat{\lambda},N(t)$ &
		Macroscopic intensity, its estimate, and total population \\
		$S,g;\ t_q,\Delta t_{\mathrm{upd}}$ &
		Source/sink, boundary flux; update time and interval \\
		$\bm v,\bm v^*,\phi$ &
		Admissible velocity, selected velocity, and scalar potential \\
		$\mathcal V_t,L^2_{\lambda}(\Omega)^d$ &
		Weighted admissible set and velocity space \\
		$\bm J,\mathcal D,\bm C$ &
		Information flux, velocity divergence, and centroid \\
		$\Gamma,\bm n_{\Gamma},J_{\Gamma}$ &
		Oriented interface, unit normal, and signed transport rate \\
		$\lambda_0,\delta\lambda,\phi_0,\delta\phi$ &
		Background fields and their first-order perturbations \\
		$\mathcal L_0,G$ &
		Background weighted operator and Green's function \\
		$\psi_n,\nu_n$ &
		Eigenfunctions and eigenvalues of $\mathcal L_0$ \\
		$R,\mathcal A$ &
		Mean spectral efficiency and area spectral efficiency \\
		\bottomrule
	\end{tabularx}
	\vspace*{-1mm}
\end{table}
Throughout the paper, bold lower-case letters denote vectors, $\|\cdot\|$ denotes the Euclidean norm, and $\nabla$, $\nabla\cdot$, and $\nabla^2$ denote the gradient, divergence, and Laplacian, respectively. The brackets $\langle\cdot,\cdot\rangle_{\Omega}$ and $\langle\cdot,\cdot\rangle_{\partial\Omega}$ denote the volume and boundary duality pairings, respectively, and $1$ denotes the constant test function equal to one. When both arguments are integrable, these pairings reduce to the corresponding volume and boundary integrals. The superscript $^*$ identifies the selected minimum-energy field, whereas $\widehat{\cdot}$ denotes an observation-based or numerical estimate. The symbol $\delta$ denotes a first-order perturbation, except in the standard Dirac notation $\delta_{\bm x}$ and $\delta(\cdot)$, whereas $\Delta$ applied to a field denotes an estimated-minus-reference discrepancy. In the variational formulation, $\bm v$ denotes a generic admissible field; reconstructed kinematic quantities use $\bm v^*$ after Proposition~\ref{Prop_Optimal_Transport}. The indicator function is denoted by $\mathbb I(\cdot)$. Symbols used only within one subsection are defined at their first occurrence. Table~\ref{tab:key_notations} summarizes the principal notation used throughout the paper.

	\section{System Model}\label{sec:system_model} 
	
We introduce the F-STSG framework. Adopting a field-theoretic perspective, this framework models the network as a coupled system comprising two distinct components: a static discrete random measure representing the network infrastructure, and a continuous intensity field representing the hydrodynamic limit of the mobile-user population.
	
	\subsection{Network Skeleton: Discrete Random Measures}\label{sec:model_skeleton} 
	
	The network infrastructure is modeled as a discrete random measure $\Phi_{\mathrm{B}}$ on the Euclidean domain $\mathbb{R}^d$ \cite{Daley2003}:
	\begin{equation} 
		\Phi_{\mathrm{B}} = \sum_{i} \delta_{\bm{X}_i},
	\end{equation}
	where $\{\bm{X}_i\}$ denotes the random locations of the base stations (BSs), and $\delta_{\bm{x}}$ is the Dirac measure concentrated at $\bm{x}$. In a $K$-tier heterogeneous network (HetNet), $\Phi_{\mathrm B}$ is the superposition of $K$ independent tier-specific point-process measures, i.e.,
	$
	\Phi_{\mathrm B}=\sum_{k=1}^{K}\Phi_k.
	$
	While the numerical validation in Section~\ref{sec:simulation} assumes a Poisson point process (PPP) for tractability, the core kinematic theory developed in Section~\ref{sec:core_theory} applies to general simple, stationary ergodic point processes.
	
	A mobile user at location $\bm{x} \in \mathbb{R}^d$ associates with the infrastructure based on the maximum biased received power. 
	Let $\mathcal{K}(i) \in \{1, \dots, K\}$ denote the tier index of the BS at $\bm{X}_i$. 
	We adopt the power-law path-loss model $r^{-\alpha_k}$, where $\alpha_k>2$ is the path-loss exponent for tier $k$.
	The association cell $V_i$ of a generic BS $\bm{X}_i$ is defined by the weighted Voronoi tessellation:
	\begin{align}  
		V_i = \Big\{\bm{x} \in \mathbb{R}^d : P_{\mathcal{K}(i)} \mathcal{W}_{\mathcal{K}(i)} \|\bm{x} - \bm{X}_i\|^{-\alpha_{\mathcal{K}(i)}} \ge \nonumber\\
		P_{\mathcal{K}(j)} \mathcal{W}_{\mathcal{K}(j)} \|\bm{x} - \bm{X}_j\|^{-\alpha_{\mathcal{K}(j)}}, \forall j \neq i \Big\},
		\label{eq:assoc_rule_generic}
	\end{align}
	where $P_k$ and $\mathcal{W}_k$ denote the transmit power and association bias weight for tier $k$, respectively.
		\begin{figure}[!t]
		\centering
		\includegraphics[width=0.93\columnwidth]{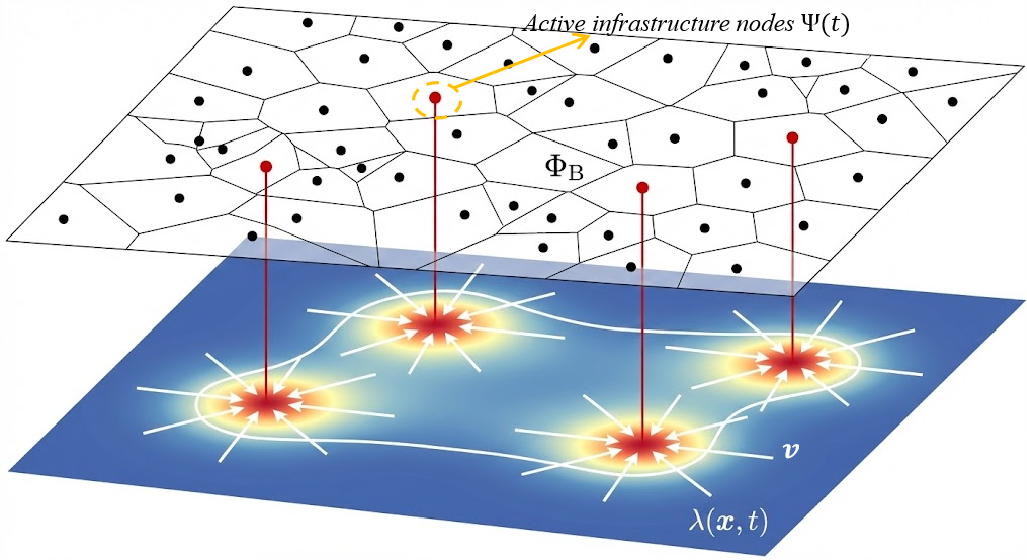}
		\vspace*{-2mm}
		\caption{Illustration of the field-measure coupling in F-STSG.}
		\label{fig:system_model} 
		\vspace*{-4mm}
	\end{figure}	
	
F-STSG represents the mobile-user population by a continuum field under the assumed scaling limit, while retaining the BS infrastructure as a discrete random measure. Thus, $\Phi_{\mathrm B}$ represents the discrete infrastructure, whereas $\lambda(\bm x,t)$ and $\bm J(\bm x,t)$ describe the evolving traffic field and its macroscopic transport. This coupling allows transport quantities such as user flux across cell boundaries to be interpreted relative to the discrete network infrastructure.

\subsection{Continuum Limit and Macroscopic Intensity Field}\label{sec:model_user_population} 
To characterize the macroscopic transport of traffic demand, let $\{\Phi_{\mathrm M}^{(n)}(t)\}_{n\geq1}$ be a sequence of mobile-user point processes indexed by the scaling parameter $n$. The associated normalized empirical measure is
\begin{equation} 
	\mu_n(A,t)=\frac{1}{n}\sum_{\bm Y_j\in\Phi_{\mathrm M}^{(n)}(t)}\mathbb I(\bm Y_j\in A),\qquad A\in\mathcal B(\mathbb R^d),
\end{equation}
where $\bm Y_j$ is the location of the $j$-th user, $\mathcal B(\mathbb R^d)$ denotes the Borel $\sigma$-algebra, and $n$ is the population-scaling parameter. We use a normalization for which $\mu_n(\mathbb R^d,t)\to1$ in probability at each fixed time. We assume that, as $n\to\infty$, $\mu_n$ converges weakly in probability to a deterministic probability density $p(\bm x,t)$:
\begin{equation} 
	\mu_n(\cdot,t)\xrightarrow{w}p(\bm x,t)\,\mathrm d\bm x\triangleq\frac{\lambda(\bm x,t)}{N(t)}\,\mathrm d\bm x,
	\label{eq:weak_convergence}
\end{equation}
where $\bm x\in\mathbb R^d$ denotes spatial location and $t\in\mathbb R^+$ denotes time, $N(t)=\int_{\mathbb R^d}\lambda(\bm x,t)\,\mathrm d\bm x$ is the total expected population, and $\lambda(\bm x,t)$ is the resulting macroscopic mobile-node intensity field. Under homogeneous per-user demand, it is proportional to the traffic-demand density; heterogeneous demand requires a correspondingly weighted intensity.

\begin{Definition}\label{def:digital_tide} 
	The macroscopic intensity field $\lambda(\bm{x},t):\mathbb{R}^d\times\mathbb{R}^+\to\mathbb{R}^+$ is the nonnegative continuum field associated with the limit in \eqref{eq:weak_convergence}, with normalized spatial density $p(\bm x,t)=\lambda(\bm x,t)/N(t)$. We use the term ``Digital Tide'' to describe the structured spatiotemporal evolution of this field.
\end{Definition}

According to the displacement theorem of SG \cite{1995Stochastic, Baccelli2009}, when points of an initial Poisson process undergo admissible independent displacements driven by a prescribed mapping, the resulting process remains Poisson with a transformed intensity. This provides one tractable microscopic realization of the macroscopic intensity evolution. The inverse field reconstruction below requires an estimated intensity and a balance law, but it does not require every operational snapshot to follow an inhomogeneous Poisson point process. In the corresponding Poisson setting, the discrete realization at time $t$, denoted by $\Phi_{\mathrm M}(t)$, can be modeled as an inhomogeneous Poisson point process governed by the advected intensity. For any Borel set $A \subseteq\mathbb{R}^d$, the expected number of users $N_A(t)$ in $A$ is
\begin{equation} 
	\mathbb{E}[N_A(t)] = \int_A \lambda(\bm{x},t) \,\mathrm{d}\bm{x}.
\end{equation}
The macroscopic intensity field obeys a balance law driven by a macroscopic velocity field $\bm{v}(\bm{x},t)$, a local source or sink $S(\bm{x},t)$, and possible boundary exchange. The source-free closed model used for the main analytical development is the special case in which local density changes are attributable to flux divergence alone.
	
	\begin{remark} 
		The continuum approximation in \eqref{eq:weak_convergence} is asymptotic and scale dependent. Under the Poisson sampling and smoothing assumptions used in Section~\ref{sec:sim_consistency}, relative fluctuations decrease with the effective number of observations in a smoothing support. For a bin or compact kernel support $B_h(\bm x)$ of scale $h$, a useful local diagnostic is
		\begin{equation}
		N_h^{\mathrm{eff}}(\bm x,t)\triangleq\int_{B_h(\bm x)}\lambda(\bm y,t)\,\mathrm d\bm y,
		\label{eq:local_effective_count}
		\end{equation}
		which is the expected number of observations contributing locally under the Poisson model. A large total population $N(t)$ does not prevent poor reconstruction where $N_h^{\mathrm{eff}}$ is small. Eq.~\eqref{eq:local_effective_count} is therefore a diagnostic variable, not a universal validity threshold: sparse rural-IoT regimes may require coarser bins, stronger regularization, a parametric density model, or a hybrid discrete--continuous description.
	\end{remark}

	\subsection{Observation of the Macroscopic Intensity Field}\label{sec:observation_model}
	The inverse problem uses a regularized field estimate rather than assuming direct access to a continuous $\lambda(\bm x,t)$. Let $t_q$ denote the $q$th field-update time. Aggregate location samples, occupancy counts, or load telemetry can be mapped to $\widehat{\lambda}(\bm x,t_q)$ by a documented spatial estimator, and neighboring regularized field estimates provide $\widehat{\partial_t\lambda}(\bm x,t_q)$. External ingress--egress measurements across $\partial\Omega$ may inform the boundary-flux estimate $\widehat g$, whereas handover counts across internal interfaces may inform or validate internal flux integrals and do not enter the outer-boundary condition. The estimated intensity must remain positive on the modeled support through a documented domain restriction or intensity floor. Source and boundary terms must be observed, modeled, or set to zero because they are not separately identifiable from density evolution alone. Empirical validation of the continuum approximation requires mobility or occupancy records aligned with the spatial grid and reporting interval. Validation of the reconstructed transport further requires trajectory or boundary-crossing measurements that are withheld from the reconstruction. The count-based observation models and their sensitivity are evaluated in Section~\ref{sec:sim_sensitivity}, while their use in constrained centroid control is evaluated in Section~\ref{sec:applications}.

	\subsection{Topological Dynamics: Non-Separable Intensity Fields}\label{sec:model_engine} 
	
	The structural properties of the intensity field $\lambda(\bm{x},t)$ characterize how its spatial profile evolves over time. To distinguish shape-preserving from shape-changing evolution, we classify the field according to the coupling between its spatial geometry and temporal variation.

		\begin{Definition}\label{def:separability}
			A strictly positive spatiotemporal field is separable if it admits the multiplicative factorization $\lambda(\bm{x},t)=a(\bm{x})b(t)$. This implies that the spatial gradient of its temporal logarithmic derivative vanishes identically over the entire domain:
			\begin{equation} 
				\nabla_{\bm{x}} \left( \frac{\partial}{\partial t} \ln \lambda(\bm{x},t) \right) \equiv \bm{0}, \quad \forall (\bm{x},t) \in \mathbb{R}^d \times \mathbb{R}^+.
				\label{eq:sep_condition}
			\end{equation}
			Conversely, a field is non-separable if this mixed derivative is not identically zero (i.e., $\nabla_{\bm{x}} \partial_t \ln \lambda \not\equiv \bm{0}$). 
		\end{Definition}
		
		A separable field preserves its spatial shape and permits only a global amplitude change. In the closed source-free setting with constant total mass, this amplitude is constant; hence, nontrivial spatial redistribution requires a non-separable field.
		
	As an example, consider a time-varying power-law topology modeling a dynamic traffic profile.  At an illustrative level, this profile can represent a localized traffic-demand hotspot, such as transient concentration around a transport hub, event venue, or dense Internet of Things (IoT) cluster, whose spatial footprint evolves over time: a larger $\beta(t)$ corresponds to stronger central concentration, whereas a smaller $\beta(t)$ produces a broader spatial distribution. The intensity is a radially symmetric field with a time-varying decay exponent $\beta(t)$:
	\begin{equation}
		\lambda(\bm{x},t) = \Lambda_0(t) (d_0 + \|\bm{x}\|)^{-\beta(t)},
	\end{equation}
	where $d_0>0$ ensures finiteness at the origin, $\beta(t)>d$ ensures integrability over $\mathbb R^d$, and $\Lambda_0(t)$ is a time-dependent normalization factor determined by the population-normalization condition
	$\int_{\mathbb R^d}\lambda(\bm{x},t)\,\mathrm d\bm{x}=N(t)$. Applying the non-separability criterion, the mixed derivative yields
		\begin{equation} 
			\nabla_{\bm{x}} \left( \frac{\partial}{\partial t} \ln \lambda(\bm{x},t) \right) = - \frac{\dot{\beta}(t)}{d_0 + \|\bm{x}\|} \frac{\bm{x}}{\|\bm{x}\|}.
		\end{equation}
		For any $\bm{x}\neq\bm 0$, this vector field is non-vanishing for $\dot{\beta}(t)\neq0$. Thus, a time-varying decay exponent changes the spatial shape of the intensity field. In the closed source-free model, such non-separable evolution requires nontrivial macroscopic transport.
	
\subsection{Dynamic State Process: Fluid-Infrastructure Coupling}\label{sec:model_coupling} 
	
		To represent the interaction between aggregate demand and network resources, F-STSG couples the continuous intensity field $\lambda(\bm{x},t)$ with the discrete infrastructure $\Phi_{\mathrm{B}}$. Let $\Psi(t)=\sum_i A_i(t)\delta_{\bm X_i}$,
		where $A_i(t)\in\{0,1\}$ indicates whether BS $i$ is active at time $t$. Thus, $\Psi(t)$ is the active-infrastructure random measure obtained by thinning $\Phi_{\mathrm B}$. In a $K$-tier HetNet, $\Psi(t)=\sum_{k=1}^{K}\Psi_k(t)$, where $\Psi_k(t)$ denotes the active thinning of tier $k$. The activation of on-demand infrastructure may depend on the local intensity $\lambda(\bm x,t)$, thereby coupling the continuous demand field to the discrete infrastructure measure.
		
		Given a realization of the active set $\Psi(t)$, a user at $\bm{x}$ associates with the serving BS $\bm{X}^* \in \Psi(t)$ providing the maximum biased received power. Consistent with the association policy in Section~\ref{sec:model_skeleton}, the serving BS is determined by
		\begin{equation} 
			\bm{X}^* = \arg \max_{\bm{X}_i \in \operatorname{supp}\Psi(t)} P_{\mathcal{K}(i)} \mathcal{W}_{\mathcal{K}(i)} \|\bm{x} - \bm{X}_i\|^{-\alpha_{\mathcal{K}(i)}}.
			\label{eq:assoc_rule}
		\end{equation}
	This dynamic association rule retains tier prioritization while allowing the active infrastructure to vary with the macroscopic state.
	
	The inverse reconstruction in Section~\ref{sec:core_theory} operates on the continuum field, while the discrete infrastructure re-enters when the reconstructed transport is interpreted relative to cell boundaries and active resources.
	\section{Field Reconstruction: The Inverse Boundary Value Problem}\label{sec:core_theory} 
	
	Under the continuum-limit assumption introduced in Section~\ref{sec:system_model}, the mobile-user distribution is represented at the selected resolution by the continuous intensity field $\lambda(\bm{x},t)$. While $\lambda(\bm{x},t)$ is the aggregate state, the underlying driver of network evolution includes macroscopic transport, local generation or depletion, and boundary exchange. We formulate the identifiable part of these latent dynamics as an inverse problem: reconstructing the minimum-energy macroscopic velocity field from the estimated density evolution together with specified source and boundary data.
	
	\subsection{Conservation Laws and Underdetermined System}\label{sec:continuity_equation} 
	
	The fundamental link between topological evolution and transport dynamics is  mass conservation. Let $S(\bm{x},t)$ denote the local rate of demand generation or depletion. For a conservative network with a vanishing source term ($S(\bm{x},t)=0$), the local rate of change in intensity is balanced by the divergence of the flux vector:
		\begin{equation} 
			\frac{\partial\lambda(\bm x,t)}{\partial t}+\nabla\cdot\bigl(\lambda(\bm x,t)\bm v(\bm x,t)\bigr)=0.
			\label{eq:continuity_conservative}
	\end{equation}
	
More generally, on a bounded observation domain $\Omega$ the balance law is
\begin{equation}
\frac{\partial\lambda}{\partial t}+\nabla\cdot(\lambda\bm v)=S
\quad\text{in }\Omega,
\label{eq:continuity_general}
\end{equation}
with outward normal flux
\begin{equation}
g(\bm x,t)\triangleq\lambda(\bm x,t)\bm v(\bm x,t)\cdot\bm n
\quad\text{on }\partial\Omega.
\label{eq:boundary_flux_general}
\end{equation}
Integration over $\Omega$ gives the compatibility identity
\begin{equation}
\frac{\mathrm d}{\mathrm dt}\int_{\Omega}\lambda\,\mathrm d\bm x
=\int_{\Omega}S\,\mathrm d\bm x-\int_{\partial\Omega}g\,\mathrm dS.
\label{eq:mass_compatibility_general}
\end{equation}
The conservative closed model in \eqref{eq:continuity_conservative} is recovered by setting $S=0$ and $g=0$.

	Assuming homogeneous per-user traffic demand, $\lambda(\bm{x},t)$ serves as the macroscopic descriptor for the spatial load. In the closed source-free setting, continuous user trajectories lead to the conservative continuity equation \eqref{eq:continuity_conservative}. While total load dynamics are driven by both local generation ($S \neq 0$) and spatial advection ($-\nabla \cdot (\lambda \bm{v})$), the closed conservative model is used first to isolate transport-induced topology reconfiguration. The extension to known $S$ and nonzero boundary flux $g$ is given below. If either term is unknown, its contribution is confounded with advection and the inverse problem is not identifiable from $\lambda$ alone.
	
	Mathematically, \eqref{eq:continuity_conservative}  provides only a single scalar constraint for the $d$-dimensional vector field $\bm{v}(\bm{x},t)$. Consequently, the inverse problem of determining $\bm v$ from $\lambda$ is non-unique up to a weighted-solenoidal component $\bm v_{\mathrm{sol}}$ satisfying $\nabla\cdot(\lambda\bm v_{\mathrm{sol}})=0$.
	
	\begin{remark} 
		It is important to differentiate F-STSG from classical vehicular traffic  models, such as the Lighthill--Whitham--Richards (LWR) model. First, traffic theory typically solves a forward problem, prescribing a velocity-density relation to predict density evolution. In contrast, F-STSG solves the inverse problem, treating the observable intensity evolution $\lambda(\bm{x},t)$ as the input to reconstruct the latent velocity field. Second, while traffic models focus solely on fluid hydrodynamics, F-STSG couples this continuous field with the discrete network infrastructure. This coupling places boundary flux and centroid drift in the same spatial coordinate system as the discrete infrastructure, so that these quantities may be used as inputs to downstream network-control models.
	\end{remark}
	
The underdetermination has two levels. Even when $S$ and $g$ are known, one may add a weighted-solenoidal component $\bm v_{\mathrm{sol}}$ satisfying $\nabla\cdot(\lambda\bm v_{\mathrm{sol}})=0$ and zero normal flux without changing the observed density. If $S$ or $g$ is unknown, source activity and transport are additionally confounded. The minimum-energy principle below resolves the first ambiguity by selecting one representative; it cannot resolve the second without supplementary observations or a source model.

\subsection{Weighted Minimum-Energy Selection and Its Optimal-Transport Connection}

	We resolve this non-uniqueness by selecting, among the admissible fields, the one with minimum instantaneous weighted kinetic energy. The objective has the same kinetic-energy form as the integrand in dynamic optimal transport, whereas uniqueness is established directly in the weighted admissible space below.

For a fixed time, let $L^2_{\lambda}(\Omega)^d$ be the weighted Hilbert space with norm $\|\bm v\|_{L^2_{\lambda}}^2=\int_{\Omega}\lambda\|\bm v\|^2\,\mathrm d\bm x$. For the closed source-free case, define the weak admissible set
\begin{equation}
\begin{aligned}
\mathcal V_t\triangleq\bigg\{\bm v\in L^2_{\lambda}(\Omega)^d:\;&
\int_{\Omega}\lambda\bm v\cdot\nabla w\,\mathrm d\bm x\\
&=\left\langle\partial_t\lambda,w\right\rangle_{\Omega},
\quad \forall w\in H^1(\Omega)\bigg\}.
\end{aligned}
\label{eq:admissible_velocity_set}
\end{equation}
Here, $H^1(\Omega)$ denotes the standard first-order Sobolev space, $H^{-1}(\Omega)$ its dual space, and $H^{-1/2}(\partial\Omega)$ the corresponding dual trace space used below for boundary data. For sufficiently smooth fields, this weak constraint is equivalent to \eqref{eq:continuity_conservative} with zero normal flux.

	\begin{Proposition}\label{Prop_Optimal_Transport} 
		Let $\Omega$ be a bounded connected Lipschitz domain, let $0<\lambda_{\min}\leq\lambda(\bm{x},t)\leq\lambda_{\max}<\infty$ almost everywhere, and let $\partial_t\lambda\in H^{-1}(\Omega)$ have zero mean. Among all fields in $\mathcal V_t$, the irrotational field $\bm{v}^*=-\nabla\phi$ is the unique minimizer of the instantaneous weighted kinetic energy
		\begin{equation} 
			E(t) = \int_{\Omega} \lambda(\bm{x},t) \|\bm{v}(\bm{x},t)\|^2 \,\mathrm{d}\bm{x}.
		\end{equation}
		Here, $\phi$ is the zero-mean weak solution of the Neumann problem in Theorem~\ref{Theorem1}. This is a pointwise-in-time minimum-energy selection. It uses the kinetic-energy integrand appearing in the Benamou--Brenier action, but it does not assert that the complete observed density path is a global Wasserstein geodesic.
	\end{Proposition}
	\begin{proof}
Theorem~\ref{Theorem1} establishes a weak potential $\phi$, unique after fixing the zero-mean gauge. With $\bm v^*=-\nabla\phi$, the weak form of \eqref{eq:poisson_general} gives
\begin{equation}
\int_{\Omega}\lambda\bm v^*\cdot\nabla w\,\mathrm d\bm x
=\left\langle\partial_t\lambda,w\right\rangle_{\Omega},
\qquad \forall w\in H^1(\Omega),
\end{equation}
so $\bm v^*\in\mathcal V_t$. For any other $\bm v\in\mathcal V_t$, let $\bm h=\bm v-\bm v^*$. Subtracting the two weak constraints yields
\begin{equation}
\int_{\Omega}\lambda\bm h\cdot\nabla w\,\mathrm d\bm x=0,
\qquad \forall w\in H^1(\Omega).
\end{equation}
Taking $w=\phi$ gives $\int_{\Omega}\lambda\bm v^*\cdot\bm h\,\mathrm d\bm x=0$. Hence,
\begin{equation}
\|\bm v\|_{L^2_{\lambda}}^2
=\|\bm v^*\|_{L^2_{\lambda}}^2
+\|\bm h\|_{L^2_{\lambda}}^2.
\label{eq:weighted_energy_decomposition}
\end{equation}
Therefore $\bm v^*$ is the unique minimum-energy element of $\mathcal V_t$. This weighted orthogonality argument gives the pointwise-in-time minimizer of the kinetic-energy integrand used in the Benamou--Brenier formulation \cite{Benamou2000}; no finite-horizon geodesic claim is required.
	\end{proof}
	
\begin{remark}
	The reconstructed field $\bm v^*$ is a macroscopic transport representative rather than a reconstruction of individual user velocities. If a finer-scale field is written as $\bm v_{\mathrm{true}}=\bm v^*+\bm v_{\mathrm{sol}}$, where $\nabla\cdot(\lambda\bm v_{\mathrm{sol}})=0$ and $\bm v_{\mathrm{sol}}$ has zero normal flux, then $\bm v_{\mathrm{sol}}$ does not alter the observed density evolution and is not identifiable from $\lambda$ alone. It may nevertheless affect microscopic mobility or channel performance, whose quantitative evaluation requires additional trajectory or channel information.
\end{remark}
	
	Substituting the optimal solution $\bm{v}^* = -\nabla \phi$ into \eqref{eq:continuity_conservative}, we derive the governing Poisson-like PDE for the flow potential:
	\begin{equation} 
		\frac{\partial \lambda(\bm{x},t)}{\partial t} = \nabla \cdot \left( \lambda(\bm{x},t) \nabla \phi(\bm{x},t) \right).
		\label{eq:conservative_pde}
	\end{equation}
	This equation forms the theoretical core of this paper, establishing a deterministic link between the observable density evolution and the latent scalar potential.

For known source and boundary data, the same minimum-energy construction leads to
\begin{align}
	\nabla\cdot(\lambda\nabla\phi)&=\partial_t\lambda-S
	&&\text{in }\Omega,\label{eq:potential_general_source}\\
	(\lambda\nabla\phi)\cdot\bm n&=-g
	&&\text{on }\partial\Omega.
	\label{eq:potential_general_boundary}
\end{align}
The boundary condition is understood in the weak conormal sense.
For this case, let $\mathcal V_t^{S,g}$ denote all $\bm v\in L^2_{\lambda}(\Omega)^d$ satisfying
\begin{equation}
\begin{aligned}
\int_{\Omega} \! \lambda \bm v \! \cdot \! \nabla w\,\mathrm d\bm x
\!=\!
\langle\partial_t\lambda \!-\! S,w\rangle_{\Omega}
\!+\!
\langle g,w\rangle_{\partial\Omega},
\forall w \! \in \! H^1(\Omega).
\end{aligned}
\label{eq:admissible_velocity_general}
\end{equation}
Under the mass-compatibility condition \eqref{eq:mass_compatibility_general}, the same weighted-orthogonality argument applies, so $\bm v^*=-\nabla\phi$ is the unique minimum-energy element. Omitting a nonzero $S$ or $g$ perturbs the forcing or boundary functional but not the operator.

	
	\subsection{Existence and Uniqueness of Scalar Potential}\label{sec:analytical_solution} 
	
	 Equation \eqref{eq:conservative_pde} constitutes a linear elliptic PDE for the scalar potential $\phi(\bm{x},t)$. We now establish the mathematical well-posedness of this field reconstruction.
	
	\begin{Theorem}\label{Theorem1}
Let $\Omega\subset\mathbb R^d$ be a bounded connected Lipschitz domain. At a fixed time, assume $\lambda\in L^{\infty}(\Omega)$ satisfies $0<\lambda_{\min}\leq\lambda\leq\lambda_{\max}<\infty$ almost everywhere, and let $\partial_t\lambda\in H^{-1}(\Omega)$ satisfy $\langle\partial_t\lambda,1\rangle_{\Omega}=0$. Consider
\begin{align}
	\nabla\cdot(\lambda\nabla\phi)&=\partial_t\lambda
	&&\text{in }\Omega,\label{eq:poisson_general}\\
	(\lambda\nabla\phi)\cdot\bm n&=0
	&&\text{on }\partial\Omega,
	\label{eq:neumann_bc}
\end{align}
where the homogeneous Neumann condition is interpreted in the same weak conormal sense as above.
Then there exists a weak solution $\phi\in H^1(\Omega)$ unique up to an additive constant. Consequently, the selected field $\bm v^*=-\nabla\phi$ is unique in $L^2_{\lambda}(\Omega)^d$.
	\end{Theorem}

\begin{proof}
	Fix the gauge by working on $H^1_{\diamond}(\Omega)=\{u\in H^1(\Omega):\int_{\Omega}u\,\mathrm d\bm{x}=0\}$. The weak problem is
		\begin{equation}
			\int_{\Omega}\lambda\nabla\phi\cdot\nabla w\,\mathrm d\bm{x}
			=-\left\langle\partial_t\lambda,w\right\rangle_{\Omega},
			\qquad \forall w\in H^1_{\diamond}(\Omega).
			\label{eq:weak_potential_problem}
		\end{equation}
		The bilinear form is bounded and, because $\lambda\geq\lambda_{\min}>0$, coercive on $H^1_{\diamond}(\Omega)$ by the Poincar\'e inequality. The right-hand side is a bounded functional on the same space. The Lax--Milgram theorem therefore gives a unique zero-mean solution. Adding an arbitrary spatial constant yields the usual Neumann equivalence class, while $\nabla\phi$ and hence $\bm v^*$ remain unique. The zero-mean condition on $\partial_t\lambda$ is the Neumann compatibility condition obtained by testing the balance with the constant function.
\end{proof}

\noindent\textit{Pointwise-in-time scope:} Theorem~\ref{Theorem1} applies separately at each time for which the stated coefficient, forcing, compatibility, and boundary assumptions hold. It does not by itself establish temporal continuity or differentiability of $t\mapsto\phi(\cdot,t)$ or $t\mapsto\bm v^*(\cdot,t)$. Such temporal regularity requires additional assumptions on the time dependence of $\lambda$, $S$, and $g$.

	\begin{Corollary} \label{cor:open_boundary_source}
Under the same positivity and domain assumptions as Theorem~\ref{Theorem1}, let $f=\partial_t\lambda-S\in H^{-1}(\Omega)$ and $g\in H^{-1/2}(\partial\Omega)$ satisfy
\begin{equation}
\langle f,1\rangle_{\Omega}=-\langle g,1\rangle_{\partial\Omega}.
\label{eq:weak_mass_compatibility}
\end{equation}
Then \eqref{eq:potential_general_source}--\eqref{eq:potential_general_boundary} admit a weak solution unique up to an additive constant. Thus, crossings of the outer observation boundary are represented by $g\neq0$, rather than by imposing a reflecting edge. An internal cell or tracking-area interface $\Gamma\subset\Omega$ is treated separately through the flux integral in Proposition~\ref{Prop_Flux_Integral}; only crossings of $\partial\Omega$ enter $g$.
\end{Corollary}

\begin{remark}\label{rem:forcing_stability}
	Let $(\lambda,f,g)$ and $(\widehat\lambda,\widehat f,\widehat g)$ satisfy the compatibility condition, where $f=\partial_t\lambda-S$, and assume $\lambda,\widehat\lambda\geq\underline\lambda>0$ almost everywhere. Let $\phi$ and $\widehat\phi$ be the corresponding zero-mean weak solutions, and define $\Delta\lambda=\widehat\lambda-\lambda\in L^\infty(\Omega)$, $\Delta f=\widehat f-f\in L^2(\Omega)$, $\Delta g=\widehat g-g\in H^{-1/2}(\partial\Omega)$, and $\Delta\phi=\widehat\phi-\phi$. Subtracting the two weak formulations and testing with $\Delta\phi$ gives
		\begin{align}		\label{eq:forcing_stability}
			\|\nabla\Delta\phi\|_{L^2(\Omega)}
			\leq\frac{1}{\underline\lambda}\big(&C_{\mathrm P}\|\Delta f\|_{L^2(\Omega)}+C_{\mathrm{tr}}\|\Delta g\|_{H^{-1/2}(\partial\Omega)} \nonumber \\
			&+\|\Delta\lambda\|_{L^\infty(\Omega)}\|\nabla\phi\|_{L^2(\Omega)}\big).
		\end{align}
Here, $C_{\mathrm P}$ and $C_{\mathrm{tr}}$ are constants associated with the Poincar\'e and trace inequalities on $\Omega$, respectively. Temporal-differentiation and source-estimation errors enter through $\Delta f$, boundary-flux errors enter through $\Delta g$, and density-estimation errors enter through the coefficient perturbation $\Delta\lambda$. Source variations that are resolved and included in $S$ do not affect elliptic well-posedness, whereas unresolved variations may produce aliasing in $\Delta f$. If $S$ is unknown, source activity and transport are not identifiable from $\lambda$ alone.
\end{remark}

	\begin{Corollary}\label{Cor_Symmetric}  
		In the source-free, no-flux setting of Theorem~\ref{Theorem1}, for a sufficiently regular radially symmetric intensity field $\lambda(r,t)$ on a disk of radius $r_{\max}$, the selected radial velocity $v_r(r,t)$ at radius $r$ admits the closed-form integral representation:
		\begin{equation} 
			v_r(r,t) = - \frac{1}{r \lambda(r,t)} \int_0^r \frac{\partial \lambda(\rho, t)}{\partial t} \rho \,\mathrm{d}\rho.	\label{eq:v_solution_integral_form}
		\end{equation}
	\end{Corollary}
	
	\begin{proof}
		 In polar coordinates, the conservative PDE \eqref{eq:poisson_general} simplifies to $\frac{1}{r}\frac{\partial}{\partial r}(r \lambda v_r) = -\frac{\partial \lambda}{\partial t}$. 
			Integrating both sides with respect to $\rho$ from $0$ to $r$ 
			and applying the non-singular boundary condition $\lim_{r\to 0} r\lambda v_r = 0$ 
			directly yields \eqref{eq:v_solution_integral_form}. This completes the proof.
	\end{proof}
	
	Theorem~\ref{Theorem1} establishes well-posedness for the stated bounded-domain, positive-density, compatibility, and boundary assumptions; it does not cover vanishing-density regions, unknown source terms, or unrestricted microscopic velocity components. Corollary~\ref{Cor_Symmetric} further provides the  analytical basis for radially symmetric systems. Intuitively, \eqref{eq:v_solution_integral_form} states that the radial velocity is the negative of the enclosed accumulation rate, normalized by the local boundary density. Thus, net accumulation inside radius $r$ corresponds to inward transport ($v_r<0$), whereas net depletion corresponds to outward transport ($v_r>0$). This links the radial transport to the conservation of traffic demand.
	
	\subsection{Analytical Properties and Physical Interpretation}\label{sec:field_properties} 
	
	The selected field $\bm v^*(\bm x,t)$, established under the assumptions of Proposition~\ref{Prop_Optimal_Transport} and Theorem~\ref{Theorem1}, is used in the subsequent F-STSG constructions. It has three relevant properties. First, Corollary~\ref{Cor_Symmetric} provides a closed-form representation for the radial setting. Second, the field is consistent with the balance law through the relation between $\partial_t\lambda$ and $\nabla\cdot(\lambda\bm v^*)$. Third, under the stated positivity and regularity assumptions, the weak solution depends continuously on the coefficient, forcing, and boundary data, as quantified in Remark~\ref{rem:forcing_stability}. These properties establish a well-defined kinematic reconstruction but do not by themselves guarantee forecasting or control performance, which also depends on observation error, latency, prediction, and the downstream controller.
	
	This selected field $\bm{v}^*(\bm{x},t)$ serves as the basis for the perturbation analysis in Section~\ref{sec:linear_response} and the  definition of the  information flux in Section~\ref{sec:metrics_toolbox}.

\subsection{Reconstruction Workflow and Computational Requirements}\label{sec:solver_workflow}

Algorithm~\ref{alg:inverse_solver} summarizes the reconstruction performed at one update time, from aggregate-field estimation and compatibility checking to the velocity and flux fields supplied to a separate application-specific controller. The corresponding post-aggregation signaling and computational requirements are described below. Let $M_{\mathrm g}$ denote the number of spatial cells used by the reconstruction grid.
\begin{algorithm}[t]
	\small
	\caption{F-STSG field reconstruction at one update time}
	\label{alg:inverse_solver}
	\begin{algorithmic}[1]
		\Require Aggregated observations at $t_q$, domain $\Omega$, source estimate $\widehat S(\cdot,t_q)$, and boundary-flux estimate $\widehat g(\cdot,t_q)$
		\State Estimate and regularize $\widehat\lambda(\cdot,t_q)$ on the $M_{\mathrm g}$-cell spatial grid
		\State Estimate $\widehat{\partial_t\lambda}(\cdot,t_q)$ from the available regularized field estimates
		\State Evaluate the discrete mass-compatibility residual associated with \eqref{eq:mass_compatibility_general} and apply any prescribed compatibility correction
		\State Solve the weighted elliptic problem for $\widehat\phi(\cdot,t_q)$ under the gauge $\int_{\Omega}\widehat\phi(\bm x,t_q)\,\mathrm d\bm x=0$
		\State Compute $\widehat{\bm v}^{*}(\cdot,t_q)=-\nabla\widehat\phi(\cdot,t_q)$ and $\widehat{\bm J}(\cdot,t_q)=\widehat\lambda(\cdot,t_q)\widehat{\bm v}^{*}(\cdot,t_q)$
		\State Return the reconstructed fields and required application metrics to a separate controller
	\end{algorithmic}
\end{algorithm}

Let $t_q$ denote the $q$th field-update time, with $t_{q+1}-t_q=\Delta t_{\mathrm{upd}}$. The continuous-time notation therefore does not imply continuous transmission of individual trajectories. At each update, the reconstruction module receives spatially aggregated observations rather than individual trajectories.

Consider a fixed-length uncompressed report containing one density value for each of the $M_{\mathrm g}$ spatial cells, $M_S$ scalar source summaries, and $M_{\partial}$ scalar boundary-flux summaries. Let $b_{\lambda}$, $b_S$, and $b_g$ denote the numbers of bits used to represent one density, source, and boundary-flux value, respectively. The post-aggregation payload and signaling rate are
	\begin{equation}
		B_{\mathrm{upd}}
		=
		M_{\mathrm g}b_{\lambda}
		+
		M_Sb_S
		+
		M_{\partial}b_g
		+
		B_{\mathrm{meta}},
		\,
		R_{\mathrm{sig}}
		=
		\frac{B_{\mathrm{upd}}}{\Delta t_{\mathrm{upd}}},
		\label{eq:signaling_overhead}
\end{equation}
where $B_{\mathrm{upd}}$ is the number of transmitted bits per update, $B_{\mathrm{meta}}$ is the metadata payload in bits, and $R_{\mathrm{sig}}$ is the post-aggregation signaling rate in bits per unit time. The metadata may contain timestamps, region identifiers, and report-format information. The source and boundary terms are omitted when they are known to be zero or are not separately reported. Equation~\eqref{eq:signaling_overhead} does not include the cost of collecting raw observations from terminals, which depends on the telemetry architecture, reporting protocol, and local preprocessing.

In the following complexity statements, $\mathcal O(\cdot)$ denotes asymptotic scaling with respect to the indicated problem-size variables. Let $N_{\mathrm{obs}}$ be the number of observations processed at one update. Under constant-time assignment of each observation to a grid cell, histogram construction requires $\mathcal O(N_{\mathrm{obs}})$ operations. More general spatial estimators may require additional estimator-specific regularization.

Let $L_{\mathrm t}$ be the number of regularized field snapshots retained by the temporal-derivative estimator. Evaluating the derivative over all grid cells requires $\mathcal O(L_{\mathrm t}M_{\mathrm g})$ operations and $\mathcal O(L_{\mathrm t}M_{\mathrm g})$ buffer memory. Let $N_{\mathrm{nz}}$ denote the number of nonzero entries in the discretized weighted elliptic operator and let $N_{\mathrm{it}}$ denote the number of iterative-solver iterations. When each iteration is dominated by sparse matrix operations, the inverse solve requires $\mathcal O(N_{\mathrm{it}}N_{\mathrm{nz}})$ operations. Storing the sparse operator, a constant number of solver vectors, and the temporal field buffer requires
	$\mathcal O\!\left(N_{\mathrm{nz}}+(L_{\mathrm t}+1)M_{\mathrm g}\right)$
	memory.

For a fixed local discretization stencil,
	$N_{\mathrm{nz}}=\mathcal O(M_{\mathrm g})$.
	The iteration count $N_{\mathrm{it}}$ nevertheless depends on the grid resolution, operator conditioning, preconditioner, and stopping tolerance. Forming the reconstructed velocity and flux and evaluating grid-based metrics such as the centroid require an additional $\mathcal O(M_{\mathrm g})$ operations.

A reconstruction can be completed before the next update only if the combined aggregation, estimation, temporal-differentiation, and inverse-solver latency is smaller than $\Delta t_{\mathrm{upd}}$. Prediction and constrained control introduce separate application-dependent latency. As a representative wall-clock benchmark, for the $20\times20$ grid used in the pNEUMA experiment, the inverse-reconstruction stage requires $6.733$~ms per update on average, with a 95th-percentile latency of $8.715$~ms on a workstation equipped with an Intel(R) Core(TM) Ultra 7 155U CPU. The statistics are obtained from 2,095 timed executions over 419 real field updates after warm-up. The measured stage starts from the prepared aggregate field and includes compatibility correction, sparse-system construction and solution, and recovery of $\phi$, $\bm v^*$, and $\bm J$; it excludes telemetry transfer, raw-data preprocessing, field estimation, and downstream control. This workstation-level timing provides a concrete computational reference alongside the asymptotic scaling analysis, but it does not establish end-to-end real-time feasibility, which additionally depends on the telemetry architecture, preprocessing and estimation pipeline, hardware platform, and update interval.
	
	
	\section{Perturbation Analysis and Green's Function Formalism}\label{sec:linear_response} 
	
	 Theorem~\ref{Theorem1} establishes the well-posed inverse problem, but spatially varying coefficients generally preclude closed-form solutions. We therefore derive a first-order approximation for a positive background field in the weak-variation regime $\epsilon\ll1$. The approximation is used only within this regime; outside it, the full weighted elliptic problem is solved directly.
	
\subsection{Linearized Field Equation}\label{sec:linearization} 

We decompose the intensity field $\lambda(\bm x,t)$ into a specified background state $\lambda_0(\bm x,t)$ and a mass-conserving first-order perturbation $\delta\lambda(\bm x,t)$ satisfying $\int_{\Omega}\delta\lambda(\bm x,t)\,\mathrm d\bm x=0$:
\begin{align} 
	\lambda(\bm x,t)&=\lambda_0(\bm x,t)+\epsilon\delta\lambda(\bm x,t),\label{eq:perturb_lambda}\\
	\phi(\bm x,t)&=\phi_0(\bm x,t)+\epsilon\delta\phi(\bm x,t),\label{eq:perturb_phi}
\end{align}
where $\epsilon\ll1$ is the perturbation parameter. The background pair $(\lambda_0,\phi_0)$ satisfies
\begin{equation} 
	\frac{\partial\lambda_0}{\partial t}=\nabla\cdot(\lambda_0\nabla\phi_0).
	\label{eq:zeroth_order}
\end{equation}

\begin{Theorem}\label{Theorem2} 
	Substituting \eqref{eq:perturb_lambda} and \eqref{eq:perturb_phi} into \eqref{eq:conservative_pde} and retaining the $\mathcal O(\epsilon)$ terms gives
	\begin{equation} 
		\nabla\cdot(\lambda_0\nabla\delta\phi)
		=\frac{\partial\delta\lambda}{\partial t}
		-\nabla\cdot(\delta\lambda\nabla\phi_0).
		\label{eq:linearized_pde}
	\end{equation}
\end{Theorem}
	
	\begin{proof}
		 Substituting the asymptotic expansions $\lambda(\bm{x},t) = \lambda_0 + \epsilon \delta\lambda$ and $\phi(\bm{x},t) = \phi_0 + \epsilon \delta\phi$ into the governing PDE \eqref{eq:conservative_pde} yields:
		\begin{equation}\label{eqApB1} 
			\frac{\partial (\lambda_0 + \epsilon \delta\lambda)}{\partial t} = \nabla \cdot \left( (\lambda_0 + \epsilon \delta\lambda) \nabla (\phi_0 + \epsilon \delta\phi) \right).
		\end{equation}
		 Expanding the right-hand side (RHS) of \eqref{eqApB1} and grouping by powers of $\epsilon$, we obtain:
		\begin{align}\label{eqRHS} 
			\text{RHS} =& \nabla \cdot (\lambda_0 \nabla \phi_0) + \epsilon \left( \nabla \cdot (\lambda_0 \nabla \delta\phi) + \nabla \cdot (\delta\lambda \nabla \phi_0) \right) \nonumber \\
			& + \mathcal{O}(\epsilon^2).
		\end{align}
		 The left-hand side (LHS) expands linearly as:
		\begin{equation}\label{eqLHS} 
			\text{LHS} =  \frac{\partial \lambda_0}{\partial t} + \epsilon \frac{\partial (\delta\lambda)}{\partial t}.
		\end{equation}
		 The background state identically satisfies the zeroth-order equation $\partial_t \lambda_0 = \nabla \cdot (\lambda_0 \nabla \phi_0)$. 
			Equating the first-order $\mathcal{O}(\epsilon)$ coefficients from both sides yields:
		\begin{equation} 
			\frac{\partial (\delta\lambda)}{\partial t} = \nabla \cdot (\lambda_0 \nabla \delta\phi) + \nabla \cdot (\delta\lambda \nabla \phi_0).
		\end{equation}
		Rearranging to isolate the weighted elliptic operator $\nabla \cdot (\lambda_0 \nabla \delta\phi)$ directly produces the linearized equation \eqref{eq:linearized_pde}.
	\end{proof}
	
	 Defining the background velocity as $\bm v_0(\bm x,t)\triangleq-\nabla\phi_0(\bm x,t)$, we write the effective first-order forcing as $S_{\mathrm{eff}}\triangleq\partial_t\delta\lambda+\nabla\cdot(\delta\lambda\bm v_0)$. This term represents the first-order mass imbalance induced by the density perturbation. The left-hand side gives the weighted elliptic response of the potential perturbation under the background field $\lambda_0$. The differential operator $\mathcal{L}_0 = \nabla \cdot (\lambda_0 \nabla)$ depends on the specified positive background state, permitting the Green's-function representation under the same boundary and compatibility conditions.
	
	\subsection{Green's Function Integral Representation}\label{sec:greens_function} 
	
	 The linearity of \eqref{eq:linearized_pde} enables the formulation of a general solution via the Green's function method. Let $G(\bm{x}, \bm{y}; t)$ be the generalized Green's function for the background operator $\mathcal{L}_0 = \nabla \cdot (\lambda_0 \nabla)$ at time $t$.
	 Under Neumann boundary conditions, the standard definition $\mathcal{L}_0 G(\bm{x}, \bm{y}; t) = \delta(\bm{x} - \bm{y})$ is ill-posed due to the global mass conservation constraint. We fix the additive constant by imposing the zero-mean gauge $\int_{\Omega}G(\bm x,\bm y;t)\,\mathrm d\bm x=0$. The Neumann Green's function is then defined by
	\begin{equation} 
		\nabla_{\bm{x}} \cdot (\lambda_0(\bm{x},t) \nabla_{\bm{x}} G(\bm{x}, \bm{y}; t)) = \delta(\bm{x} - \bm{y}) - \frac{1}{|\Omega|},
	\end{equation}
	subject to $\partial_{\bm n}G=0$ on $\partial\Omega$, where $\delta(\cdot)$ denotes the Dirac delta function.
	Because the effective source $S_{\mathrm{eff}}$ satisfies $\int_{\Omega} S_{\mathrm{eff}} (\bm{y},t) \,\mathrm{d} \bm{y} = 0$, the constant offset  $-1/|\Omega|$ vanishes upon integration.
	The potential perturbation therefore has the Green's-function integral representation
	\begin{equation} 
		\delta\phi(\bm{x},t) = \int_{\Omega} G(\bm{x}, \bm{y}; t) S_{\mathrm{eff}}(\bm{y},t) \,\mathrm{d}\bm{y}.
		\label{eq:greens_solution}
	\end{equation}
 Consequently, substituting $S_{\mathrm{eff}}$, the velocity perturbation $\delta \bm{v} = -\nabla (\delta \phi)$ becomes:

	\begin{align}
		\delta\bm v(\bm x,t)
		=&-\nabla_{\bm x}\int_{\Omega}G(\bm x,\bm y;t)\Bigl[
		\frac{\partial\delta\lambda(\bm y,t)}{\partial t}\nonumber\\
		&\hspace{15mm}+\nabla_{\bm y}\cdot\bigl(\delta\lambda(\bm y,t)\bm v_0(\bm y,t)\bigr)\Bigr]\,\mathrm d\bm y.
\end{align}
	 Within the linearization regime, $G$ maps the compatible first-order effective source to the first-order velocity perturbation. Outside the small-perturbation and positive-background assumptions, the full weighted elliptic problem must be solved directly.
	
	\subsection{Spectral Decomposition of the Weighted Operator}\label{sec:spectral_properties} 
	
	 The following spectral decomposition is performed at a fixed observation time, and the time arguments of $\mathcal L_0$, $G$, $\psi_n$, and $\nu_n$ are suppressed in this subsection. The background operator $\mathcal L_0=\nabla\cdot(\lambda_0\nabla)$ acts as a generalized Laplacian and is self-adjoint on $L^2(\Omega)$ under the stated Neumann boundary condition.
	
	 According to spectral theory for elliptic operators~\cite{Courant1953}, there exists a countable set of orthonormal eigenfunctions $\{\psi_n(\bm{x})\}$ and corresponding non-negative eigenvalues, $0 = \nu_0 < \nu_1 \le \nu_2 \le \dots$, satisfying:
	\begin{equation} 
		\nabla \cdot (\lambda_0(\bm{x}) \nabla \psi_n(\bm{x})) = -\nu_n \psi_n(\bm{x}),
	\end{equation}
	subject to the Neumann boundary condition $\partial_{\bm n}\psi_n=0$ on $\partial\Omega$. The zero eigenvalue $\nu_0=0$ corresponds to the constant eigenfunction $\psi_0 = |\Omega|^{-1/2}$, which reflects the global mass conservation.
	
	 By the spectral theorem for self-adjoint elliptic operators, the Green's function on the subspace orthogonal to $\psi_0$ admits the expansion
	\begin{equation} 
		G(\bm{x}, \bm{y}) = - \sum_{n=1}^{\infty} \frac{\psi_n(\bm{x}) \psi_n(\bm{y})}{\nu_n}.
		\label{eq:spectral_expansion}
	\end{equation}
	 Substituting \eqref{eq:spectral_expansion} into \eqref{eq:greens_solution}, the potential perturbation becomes:
	
		\begin{equation}
			\delta\phi(\bm{x},t)=-\sum_{n=1}^{\infty}\frac{\psi_n(\bm{x})}{\nu_n}s_n^{\mathrm{eff}}(t),
	\end{equation}
	 where $s_n^{\mathrm{eff}}(t)\triangleq\int_{\Omega}\psi_n(\bm{y})S_{\mathrm{eff}}(\bm{y},t)\,\mathrm d\bm y$ is the projection of the effective source onto the $n$-th eigenmode.
	
	 The expansion shows how the inverse operator weights the spatial eigenmodes of the effective forcing. Each mode contributes to $\delta\phi$ with a factor $1/\nu_n$, so modes with larger $\nu_n$ have a smaller contribution than modes with smaller $\nu_n$. This is a property of the elliptic inverse operator and does not imply a temporal decay rate. This spectral property does not by itself establish robustness to observation noise; finite-count and estimator sensitivity are evaluated separately in Section~\ref{sec:sim_sensitivity}.
	
	\section{Kinematic Analytical Toolbox}\label{sec:metrics_toolbox} 
	
	 Based on the selected macroscopic velocity field $\bm{v}^*(\bm{x},t)$, we formulate a suite of kinematic metrics. These quantities describe the instantaneous transport implied by the balance law. Their use for prediction requires causal estimation and an explicit model of subsequent field evolution. For the classical forms of the differential identities and interface integrals below, we additionally assume that $\lambda$ is continuously differentiable in space and time, $\bm v^*$ is continuously differentiable in space, and the interfaces are piecewise smooth. These additional assumptions are not required for the weak existence and uniqueness result in Theorem~\ref{Theorem1}.

\vspace*{-2mm}
\subsection{Information Flux Vector}\label{sec:metric_flux}

Generalizing the concept of microscopic node kinetics, we define the information flux as a vector field quantifying the spatial advection of the intensity field.

\begin{Definition}\label{Definition4}
	The information flux, $\bm J(\bm x,t)$, is defined as the product of the local intensity and macroscopic velocity field. For reconstructed metrics, the velocity is the selected field $\bm v^*$:
	\begin{equation}
		\bm J(\bm x,t)\triangleq\lambda(\bm x,t)\bm v^*(\bm x,t).
	\end{equation}
\end{Definition}

The information-flux field describes transport of the modeled intensity measure and defines signed flow across an oriented interface $\Gamma$, such as a cell edge or tracking-area boundary.

\begin{Proposition}\label{Prop_Flux_Integral}
	Let $\Gamma\subset\Omega$ be an oriented $(d-1)$-dimensional interface with chosen unit normal $\bm n_{\Gamma}$. The signed transport rate across $\Gamma$ in the direction of $\bm n_{\Gamma}$ is
	\begin{equation}
		J_{\Gamma}(t)
		=
		\int_{\Gamma}
		\bm J(\bm x,t)\cdot\bm n_{\Gamma}\,\mathrm dS.
		\label{eq:handover_flux}
	\end{equation}
\end{Proposition}

\begin{proof}
	The result follows by integrating the normal component of $\bm J$ over the oriented interface $\Gamma$.
\end{proof}

The flux integral $J_{\Gamma}(t)$ characterizes instantaneous transport across an internal interface. For a fixed association partition $\{V_i\}$, integrating the continuity equation over $V_i\cap\Omega$ gives
\begin{equation}
	\frac{\mathrm{d}}{\mathrm{d}t}\int_{V_i\cap\Omega}\lambda\,\mathrm{d}\mathbf{x}
	=\int_{V_i\cap\Omega}S\,\mathrm{d}\mathbf{x}
	-\sum_{j\in\mathcal{N}_i}J_{ij}-J_{i,\mathrm{ext}},
	\label{eq:cell_population_balance}
\end{equation}
where $\mathcal N_i$ denotes the neighboring regions, $J_{ij}$ is the signed flux from $V_i$ to $V_j$ defined by Proposition~\ref{Prop_Flux_Integral}, and $J_{i,\mathrm{ext}}$ is the outward flux through $\partial V_i\cap\partial\Omega$.

\begin{remark}\label{rem:ase_link}
	The reconstructed transport field can also be coupled to a separately specified wireless-performance model. Let $\operatorname{SINR}(\bm x,t)$ denote the local signal-to-interference-plus-noise ratio, and let
	\[
	R(\bm x,t)
	=
	\mathbb E\!\left[
	\log_2\!\left(1+\operatorname{SINR}(\bm x,t)\right)
	\right]
	\]
	denote the local mean spectral efficiency under a specified association, activity, scheduling, channel, and interference model. When $\lambda(\bm x,t)$ represents the active-link density, the local area spectral efficiency is $\mathcal A(\bm x,t)=\lambda(\bm x,t)R(\bm x,t)$. Combining this definition with \eqref{eq:continuity_general} gives
	\[
	\frac{\partial\mathcal A}{\partial t}
	+
	\nabla\cdot(R\bm J)
	=
	RS
	+
	\lambda\frac{\partial R}{\partial t}
	+
	\bm J\cdot\nabla R.
	\]
	Thus, $R\bm J$ provides a spectral-efficiency-weighted transport flux, while $R$ itself is determined by the separate wireless-performance model. If $\lambda$ denotes demand rather than active-link density, an additional activity or scheduling map is required.
	For a compatible finer-scale field $\bm v_{\mathrm{true}}=\bm v^*+\bm v_{\mathrm{sol}}$ with $\nabla\cdot(\lambda\bm v_{\mathrm{sol}})=0$, the omitted component contributes
	\[
		\nabla\cdot(R\lambda\bm v_{\mathrm{sol}})
		-
		\lambda\bm v_{\mathrm{sol}}\cdot\nabla R
		=
		R\nabla\cdot(\lambda\bm v_{\mathrm{sol}})
		=
		0
	\]
	to the area-spectral-efficiency balance for a specified $R$. The same circulation may nevertheless alter $R$ through microscopic channel or interference dynamics; such an effect is not identifiable from $\lambda$ alone and requires additional mobility or channel information. Hence, this relation provides a kinematic interface to, rather than a replacement for, a wireless-performance model.
\end{remark}

Related fluid-kinetic formulations have also been used 	for loop-free routing in ultra-dense networks \cite{Dong2026VDR}, physics-native digital-twin synchronization in kinetic swarms \cite{DongTDT2026}, and time-varying interference modeling in LEO	mega-constellations \cite{DongLEOKIF2026}.

\vspace*{-2mm}	
\subsection{Topological Compression and Material Derivative}\label{sec:metric_divergence} 
	
	To characterize the formation of demand hotspots at the macroscopic level, we analyze the compressibility of the selected flow.
	
	\begin{Definition}\label{Definition5} 
		The congestion divergence, $\mathcal{D}(\bm{x},t)$, is the scalar field defined by the divergence of the macroscopic velocity:
		\begin{equation} 
			\mathcal{D}(\bm{x},t) \triangleq \nabla \cdot \bm{v}^{*}(\bm{x},t).
		\end{equation}
	\end{Definition}
	
	\begin{Proposition}\label{Prop_Material_Derivative} 
		For a conservative system ($S=0$), the congestion divergence is mathematically equivalent to the negative logarithmic rate of change of the intensity, as observed in the Lagrangian frame of reference:
		\begin{equation}\label{eqPro3} 
			\mathcal{D}(\bm{x},t) = - \frac{1}{\lambda(\bm{x},t)} \frac{D\lambda(\bm{x},t)}{Dt},
		\end{equation}
		where $\frac{D}{Dt} \triangleq \frac{\partial}{\partial t} + \bm{v}^{*} \cdot \nabla$ is the material (or Lagrangian) derivative, capturing the rate of change experienced by a moving fluid element.
	\end{Proposition}
	
	\begin{proof}
	 Expanding the divergence term in the conservative continuity equation \eqref{eq:continuity_conservative} using the scalar-vector product identity yields: 
			\begin{equation} 
				\frac{\partial \lambda}{\partial t} + \bm{v}^{*} \cdot \nabla \lambda + \lambda (\nabla \cdot \bm{v}^{*}) = 0.
			\end{equation}
			Applying the definition of the material derivative $\frac{D\lambda}{Dt} \triangleq \frac{\partial \lambda}{\partial t} + \bm{v}^{*} \cdot \nabla \lambda$ and the congestion divergence $\mathcal{D} \triangleq \nabla \cdot \bm{v}^{*}$, the equation simplifies to:
			\begin{equation} 
				\frac{D\lambda}{Dt} + \lambda \mathcal{D} = 0.
			\end{equation}
			Assuming $\lambda > 0$ and solving for $\mathcal{D}$ directly yields \eqref{eqPro3}.
	\end{proof}

For the general balance law \eqref{eq:continuity_general}, the same expansion gives
	\begin{equation}
		\mathcal D(\bm{x},t)
		=\frac{S(\bm{x},t)}{\lambda(\bm{x},t)}
		-\frac{1}{\lambda(\bm{x},t)}\frac{D\lambda(\bm{x},t)}{Dt}.
		\label{eq:material_derivative_source}
	\end{equation}
	Thus, $\mathcal D<0$ identifies local compression of the selected macroscopic transport field, whereas density growth may also arise from $S>0$. Forecasting or early-warning use additionally requires causal field estimation and an evolution or service model.
	
	\subsection{Global Centroid Dynamics}\label{sec:metric_centroid} 
	
	To track the global trajectory of the demand, we define the first moment of the intensity field.
	
	\begin{Definition}\label{Definition6} 
		The network centroid, $\bm{C}(t)$, is the instantaneous expectation of the user location vector:
		\begin{equation} 
			\bm{C}(t) \triangleq \mathbb{E}_{\bm{x} \sim p(\cdot,t)}[\bm{x}] = \frac{1}{N(t)} \int_{\mathbb{R}^d} \bm{x} \lambda(\bm{x},t) \,\mathrm{d}\bm{x}.
		\end{equation}
	\end{Definition}
	
	\begin{Proposition}\label{Prop_Centroid_Drift} 
		In a source-free conservative system ($S=0$) with sufficiently fast flux decay at infinity, the velocity of the centroid, $\bm{V}_C(t) \triangleq \frac{\mathrm{d}\bm{C}}{\mathrm{d}t}$, is exactly the intensity-weighted spatial average of the macroscopic velocity field:
		\begin{equation}\label{eqVoC} 
			\bm{V}_C(t) = \mathbb{E}_{\bm{x} \sim p(\cdot,t)}[\bm{v}^{*}(\bm{x},t)] = \int_{\mathbb{R}^d} \bm{v}^{*}(\bm{x},t) p(\bm{x},t) \,\mathrm{d}\bm{x}.
		\end{equation}
	\end{Proposition}
	
\begin{proof}
	Differentiating $\bm C(t)=N(t)^{-1}\int_{\mathbb R^d}\bm x\lambda(\bm x,t)\,\mathrm d\bm x$ and using $\mathrm dN(t)/\mathrm dt=0$ gives
	\begin{equation} 
		N(t)\bm V_C(t)=\int_{\mathbb R^d}\bm x\frac{\partial\lambda(\bm x,t)}{\partial t}\,\mathrm d\bm x.
	\end{equation}
	 Under the source-free balance law, substituting $\partial_t\lambda=-\nabla\cdot(\lambda\bm v^*)$ gives
	\begin{equation} 
		N(t)\bm V_C(t)=-\int_{\mathbb R^d}\bm x\nabla\cdot(\lambda\bm v^*)\,\mathrm d\bm x.
	\end{equation}
	For the $i$-th component, integration by parts gives
	\begin{equation} 
		\int_{\mathbb R^d}x_i\nabla\cdot(\lambda\bm v^*)\,\mathrm d\bm x=-\int_{\mathbb R^d}\nabla x_i\cdot(\lambda\bm v^*)\,\mathrm d\bm x,
	\end{equation}
	where the boundary term vanishes when the flux decays sufficiently rapidly at infinity. Since $\nabla x_i=\bm e_i$, where $\bm e_i$ is the $i$th Cartesian basis vector,
	\begin{equation} 
		N(t)\bm V_C(t)=\int_{\mathbb R^d}\lambda(\bm x,t)\bm v^*(\bm x,t)\,\mathrm d\bm x.
	\end{equation}
	Dividing by $N(t)$ yields \eqref{eqVoC}.
\end{proof}
	

For the general balance law on a bounded observation domain, define
$N_{\Omega}(t)=\int_{\Omega}\lambda(\bm x,t)\,\mathrm d\bm x$
and
$\bm C_{\Omega}(t)=N_{\Omega}(t)^{-1}\int_{\Omega}\bm x\lambda(\bm x,t)\,\mathrm d\bm x$.
Differentiating $N_{\Omega}(t)\bm C_{\Omega}(t)$ and using the mass-compatibility identity \eqref{eq:mass_compatibility_general} gives
	\begin{align}
			N_{\Omega}(t)\dot{\bm C}_{\Omega}(t)
			=&{}\int_{\Omega}\lambda\bm v^*\,\mathrm d\bm x
			+\int_{\Omega}\bigl(\bm x-\bm C_{\Omega}(t)\bigr)S\,\mathrm d\bm x  \nonumber\\
			&-\int_{\partial\Omega}\bigl(\bm x-\bm C_{\Omega}(t)\bigr)g\,\mathrm dS.
		\label{eq:centroid_general_balance}
\end{align}
Thus, in the source-free closed case ($S=0$ and $g=0$), $\dot{\bm C}_{\Omega}$ reduces to the intensity-weighted mean velocity. More generally, source activity and boundary exchange modify the centroid dynamics through their first-moment contributions, although these contributions may vanish or cancel in special cases.

\section{Numerical and Empirical Validation}\label{sec:simulation}
The experiments in this section serve different evidentiary purposes, summarized in Table~\ref{tab:evidence_scope}. The matched tests are used only for convergence, implementation, identifiability, or analytical-consistency checks; controlled perturbations evaluate sensitivity; the pNEUMA experiment provides held-out real-trace validation of macroscopic transport; and the UAV study evaluates a downstream control use case rather than reconstruction fidelity.

\begin{table*}[!t]
	\centering
	\caption{Evidentiary role and scope of the experiments in Section~VI.}
	\label{tab:evidence_scope}
	\vspace*{-2.5mm}
	\scriptsize
	\setlength{\tabcolsep}{2.5pt}
	\begin{tabular}{p{0.05\textwidth}p{0.24\textwidth}p{0.32\textwidth}p{0.32\textwidth}}
		\hline
		Section & Evidence type & Purpose & Evidentiary limit\\
		\hline
		VI-A & Matched finite-population convergence & Check convergence of the selected density estimator under the prescribed generator and smoothing rule. & Does not establish reconstruction accuracy on operational mobility data.\\
		VI-B & Implementation and identifiability check & Verify minimum-energy selection and separation of the density-identifiable and weighted-solenoidal components. & The $0.092\%$ normalized root-mean-square error (NRMSE) is not an unseen-data accuracy result.\\
		VI-C & Sensitivity and robustness analysis & Isolate finite-count, temporal-estimation, missing-data, source, boundary, and operator perturbations. & Uses manufactured reference fields and does not measure empirical fidelity.\\
		VI-D & Real-trace empirical validation & Compare reconstructed internal transport with held-out trajectory crossings in pNEUMA. & Evaluates macroscopic transport, not microscopic vehicle velocities.\\
		VI-E & Matched analytical-consistency check & Check the derived scalar kinematic identities and illustrate their temporal interpretation. & Does not provide independent reconstruction-accuracy evidence.\\
		VI-F & Downstream control evaluation & Evaluate the centroid as a causal reference under observation delay and motion constraints. & Does not validate the inverse reconstruction itself.\\
		\hline
	\end{tabular}
\end{table*}

	\subsection{Finite-Population Convergence}\label{sec:sim_consistency} 
	
	Section~\ref{sec:model_user_population} introduces the weak-convergence limit of the normalized empirical measure $\mu_n$ as $n\to\infty$. We examine its finite-population behavior using the Monte Carlo population size $N_{\mathrm{pop}}$ and the relative error
\begin{equation} 
\mathcal E(N_{\mathrm{pop}})
\triangleq
\frac{\left\|\widehat{\lambda}_{N_{\mathrm{pop}}}(\bm x,t)-\lambda(\bm x,t)\right\|_{L^2(\Omega)}}{\left\|\lambda(\bm x,t)\right\|_{L^2(\Omega)}}.
\label{eq:finite_population_error}
\end{equation}
Here, $\lambda(\bm x,t)$ is the theoretical macroscopic intensity, and $\widehat{\lambda}_{N_{\mathrm{pop}}}(\bm x,t)$ is the empirical density field reconstructed from a Monte Carlo realization containing $N_{\mathrm{pop}}$ users.

\begin{figure}[!t]
\vspace*{-3mm}
\centering
\includegraphics[width=0.9\columnwidth]{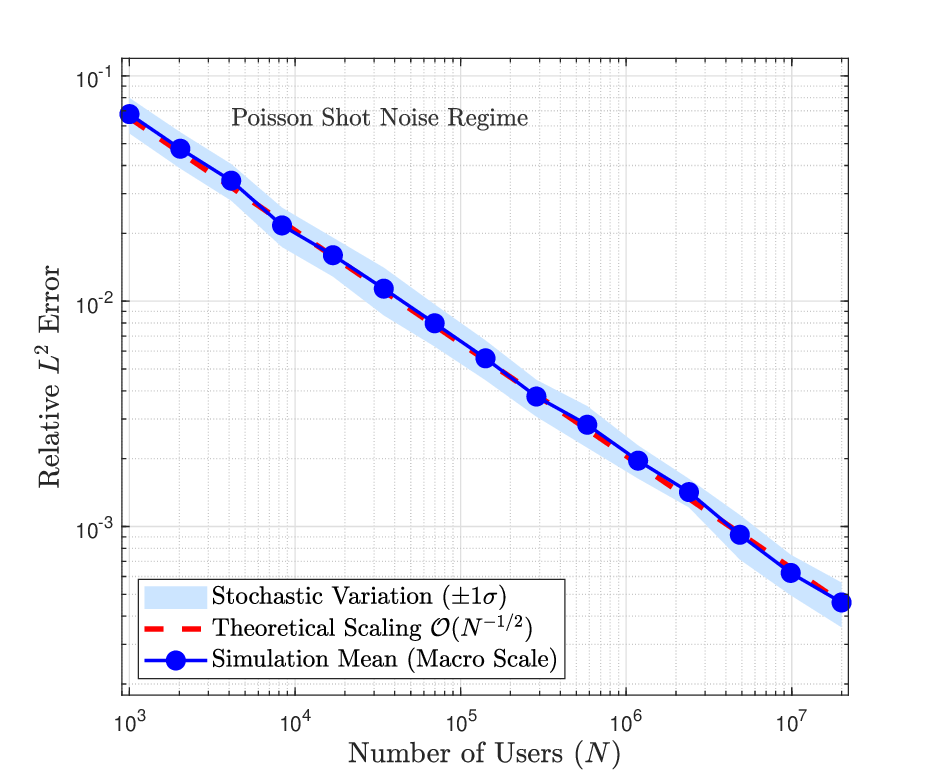}
\vspace*{-5mm}
\caption{Finite-population convergence under the matched generator. Blue markers show the simulated mean relative density-estimation error $\mathcal E(N_{\mathrm{pop}})$, the shaded band indicates $\pm1$ standard deviation across independent realizations, and the red dashed line shows the theoretical $\mathcal O(N_{\mathrm{pop}}^{-1/2})$ scaling.}
\label{fig:convergence} 
\vspace*{-4mm}
\end{figure}

We use a time-varying Gaussian hotspot to generate a non-separable intensity field. To separate sampling error from grid-integration error, the theoretical intensity $\lambda(\bm x,t)$ is evaluated by the cumulative-distribution-function difference method. For each realization $m\in\{1,\ldots,M\}$, a constellation $\Phi_m^{(N_{\mathrm{pop}})}$ is generated by drawing $N_{\mathrm{pop}}$ coordinates from the instantaneous density profile. The empirical field $\widehat{\lambda}_{N_{\mathrm{pop}}}$ is then estimated using the same smoothing rule for all population sizes. The results are averaged over $M=10^5$ independent realizations for each $N_{\mathrm{pop}}$.

Fig.~\ref{fig:convergence} shows that the relative error follows an approximately $N_{\mathrm{pop}}^{-1/2}$ trend over the tested range, while the snapshot variability decreases with population size. At the largest simulated population, $N_{\mathrm{pop}}\approx10^7$, the error is below $10^{-3}$ for the selected estimator and Poisson generator. These finite-range results are consistent with the expected asymptotic convergence under the matched generator and smoothing rule, but they do not provide empirical validation for operational mobility traces. For a measured network, the adequacy of the continuum approximation depends on the spatial aggregation scale, reporting interval, local effective observation count, and estimator bias and must be assessed on the target dataset.

	\begin{figure}[!b]
		\vspace*{-6mm}
		\centering
		\includegraphics[width=\columnwidth]{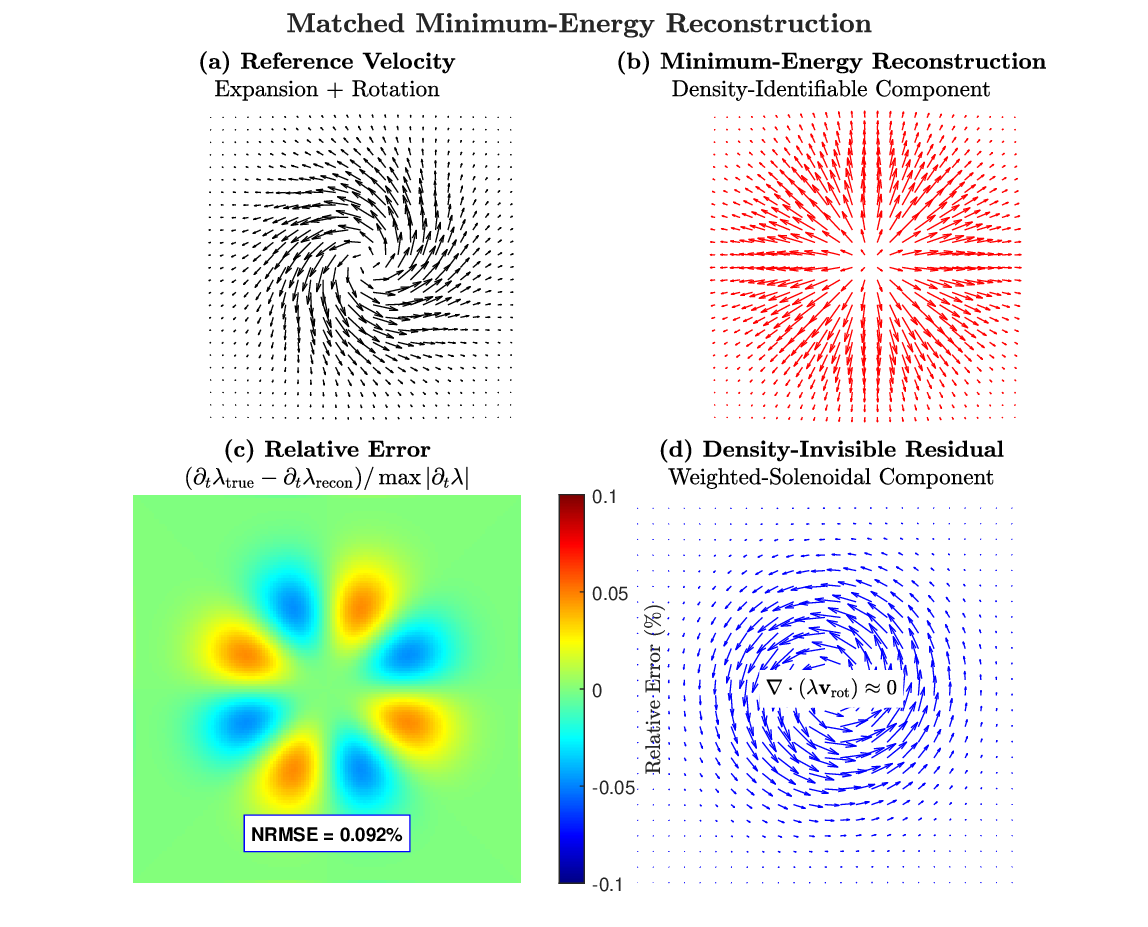}
		\vspace*{-10mm}
		\caption{Matched inverse-problem check. (a) Reference velocity $\bm v_{\mathrm{true}}$, formed by expansion and rotation. (b) Reconstructed minimum-energy field $\bm v^*$. (c) Relative reconstruction-error map. (d) Residual $\bm v_{\mathrm{true}}-\bm v^*$, corresponding to the injected density-invisible circulation.}
		\label{fig:inverse_validation} 
		\vspace*{-1mm}
	\end{figure}
	
	
	\subsection{Matched Minimum-Energy Reconstruction}\label{sec:sim_velocity}
	
	Complementing the stochastic convergence analysis in Section~\ref{sec:sim_consistency}, this subsection checks the inverse solver using deterministic manufactured continuum fields. Proposition~\ref{Prop_Optimal_Transport} shows that the F-STSG reconstruction selects the minimum-energy admissible transport component and excludes weighted-solenoidal modes that do not alter the observed density evolution.
	
	To verify this property and the resolution of the inverse problem's non-uniqueness, we construct a reference velocity field $\bm{v}_{\mathrm{true}}$ as a superposition of a compressive component and a weighted-solenoidal component:
	\begin{equation} 
		\bm{v}_{\mathrm{true}} = \bm{v}_{\mathrm{irr}} + \bm{v}_{\mathrm{rot}}.
	\end{equation}
	Here, $\bm v_{\mathrm{irr}}$ is a radial gradient field that drives the prescribed density evolution. The injected rotational component $\bm v_{\mathrm{rot}}$ is constructed to satisfy $\nabla\cdot(\lambda\bm v_{\mathrm{rot}})=0$. In the present construction, this weighted-solenoidal condition follows from $\nabla\cdot\bm v_{\mathrm{rot}}=0$ and $\bm v_{\mathrm{rot}}\cdot\nabla\lambda=0$. The rotational component contributes to the local flux and kinetic energy but does not change the density evolution.
	
The inverse solver receives only the temporal derivative $\partial_t\lambda$ generated by $\bm v_{\mathrm{true}}$. Fig.~\ref{fig:inverse_validation}(a) shows the reference field, which contains both the radial transport component and a rotational component, while Fig.~\ref{fig:inverse_validation}(b) shows that the reconstruction retains the density-identifiable radial component.  Fig.~\ref{fig:inverse_validation}(c) further shows agreement with this identifiable component, with an  NRMSE of approximately $0.092\%$, where the NRMSE is $\|\bm v^*-\bm v_{\mathrm{irr}}\|_{L^2}/\|\bm v_{\mathrm{irr}}\|_{L^2}$. Because the reference field and inverse solver use matched numerical settings, this small error is interpreted only as an implementation and identifiability check rather than reconstruction accuracy on unseen data. Finally, Fig.~\ref{fig:inverse_validation}(d) shows the remaining weighted-solenoidal component, which is invisible to the density evolution but may still affect microscopic mobility or user-level performance.

\begin{figure*}[!t]
	\vspace*{-1mm}
	\centering
	\includegraphics[width=0.9\textwidth]{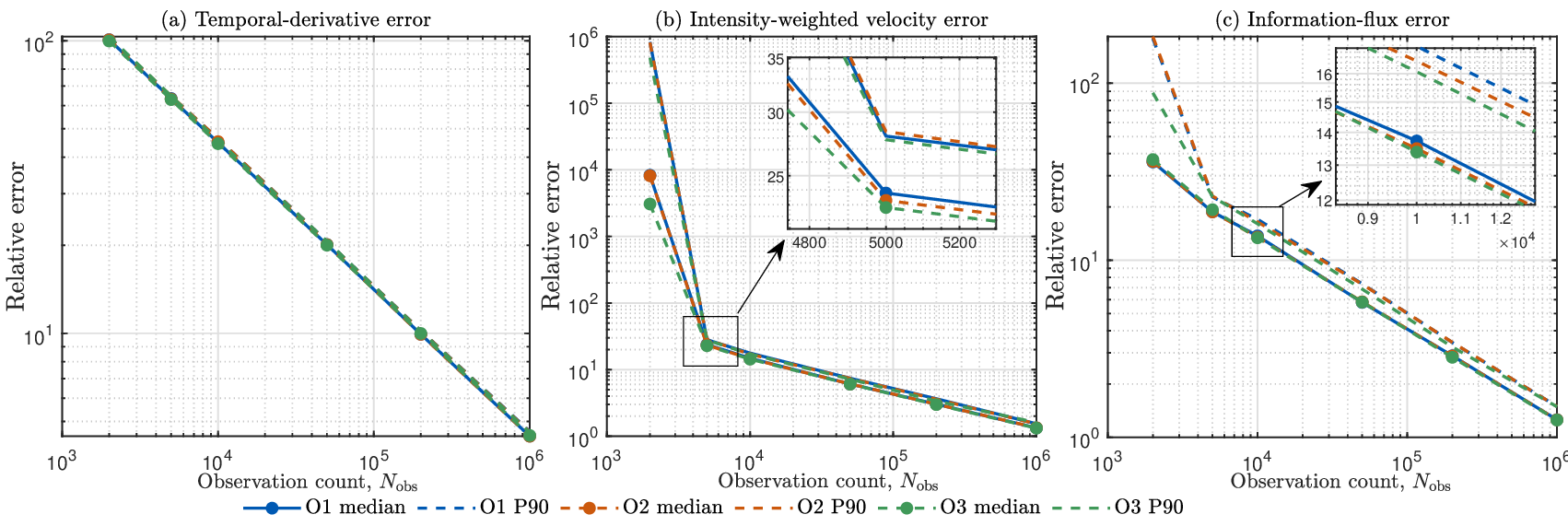}
	\vspace*{-5mm}
	\caption{Sensitivity to aggregate-count observations. $\mathrm{O1}$ uses independently sampled grid-bin Poisson counts, $\mathrm{O2}$ rescales each $\mathrm{O1}$ snapshot to the prescribed aggregate total, and $\mathrm{O3}$ uses fixed-total multinomial counts. Solid and dashed curves show the median and 90th percentile, respectively, over $30$ seeds.}
	\label{fig:observation_sensitivity}
	\vspace*{-4mm}
\end{figure*}
	
\subsection{Sensitivity to Observation, Estimation, and Model Perturbations}\label{sec:sim_sensitivity}

We next examine the robustness of the reconstruction pipeline under controlled perturbations using manufactured fields on a $96\times96$ grid and $30$ random seeds. These controlled tests complement the real-trace evaluation in Section~\ref{sec:sim_real_trace} by providing known reference fields and allowing individual error sources to be isolated. The references provide $\lambda$, $\partial_t\lambda$, $\bm v_{\mathrm{ref}}$, and $\bm J_{\mathrm{ref}}=\lambda\bm v_{\mathrm{ref}}$, while the reconstruction pipeline produces $\widehat{\lambda}$, $\widehat{\partial_t\lambda}$, $\widehat{\bm v}^{*}$, and $\widehat{\bm J}=\widehat{\lambda}\widehat{\bm v}^{*}$.

Three aggregate-count observation models are considered. In $\mathrm{O1}$, each grid-bin count is sampled independently from a Poisson distribution, so the aggregate count varies across snapshots. In $\mathrm{O2}$, each independently sampled Poisson snapshot is rescaled to the prescribed aggregate total. In $\mathrm{O3}$, the grid-bin counts are jointly sampled from a multinomial distribution with a fixed total count. Model $\mathrm{O1}$ is retained as an independent-snapshot stress test. When an $\mathrm{O1}$ realization violates the closed-domain compatibility condition, the spatial mean of the estimated forcing is removed before the Neumann solve.

The reconstruction errors are measured relative to the corresponding manufactured reference fields as
\begingroup
\begin{subequations}\label{eq:sensitivity_metrics}
	
	\begin{equation}
		e_{\partial_t\lambda}
		=
		\frac{
			\bigl\|
			\widehat{\partial_t\lambda}
			-
			\partial_t\lambda
			\bigr\|_{L^2(\Omega)}
		}{
			\bigl\|
			\partial_t\lambda
			\bigr\|_{L^2(\Omega)}
		}.
		\label{eq:sensitivity_dt}
	\end{equation}
	
	\begin{equation}
		e_{v,\lambda}
		=
		\frac{
			\bigl\|
			\widehat{\bm v}^{*}-\bm v_{\mathrm{ref}}
			\bigr\|_{L^2_{\lambda}(\Omega)}
		}{
			\bigl\|
			\bm v_{\mathrm{ref}}
			\bigr\|_{L^2_{\lambda}(\Omega)}
		}.
		\label{eq:sensitivity_velocity}
	\end{equation}
	
	\begin{equation}
		e_J
		=
		\frac{
			\bigl\|
			\widehat{\bm J}
			-
			\bm J_{\mathrm{ref}}
			\bigr\|_{L^2(\Omega)}
		}{
			\bigl\|
			\bm J_{\mathrm{ref}}
			\bigr\|_{L^2(\Omega)}
		}.
		\label{eq:sensitivity_flux}
	\end{equation}
	
\end{subequations}
\endgroup
Here, $e_{\partial_t\lambda}$ measures temporal-derivative estimation error. The velocity metric $e_{v,\lambda}$ uses the intensity-weighted norm of the minimum-energy formulation and gives greater weight to errors in high-intensity regions. The flux metric $e_J$ includes errors in both $\widehat\lambda$ and $\widehat{\bm v}^{*}$. All cases included in the reported relative-error summaries have nonzero reference norms. Medians and 90th percentiles are reported because the finite-count error distributions are skewed.

With exact fields, exact temporal derivatives, and matched discrete operators, the reconstruction errors approach numerical precision, which serves only as an algebraic implementation check. Once finite-count observations are introduced, Fig.~\ref{fig:observation_sensitivity} shows that increasing the aggregate count reduces the reconstruction error, while temporal differentiation remains the dominant error source in the tested observation settings.

The isolated source, boundary, and operator perturbations produce smaller errors than the finite-count observation cases over the tested ranges. In contrast, combining low observation counts, missing grid bins, source and boundary perturbations, and discretization mismatch can produce velocity and flux errors larger than their corresponding reference norms.

These results identify a tradeoff between spatiotemporal resolution and observation noise. Finer updates do not necessarily improve reconstruction when the available count per update becomes too small, a behavior that is also observed in the real-trace experiment below.

\subsection{Real-Trace Validation With pNEUMA Mobility Data}\label{sec:sim_real_trace}
We next test F-STSG on the public pNEUMA urban-mobility dataset~\cite{ref_pNEUMA}. Three sessions recorded on October~29, October~30, and November~1, 2018 are used, comprising 23,494 vehicle trajectories and more than $7.3\times10^{7}$ recorded positions. The trajectories are binned into spatial counts to estimate $\widehat{\lambda}(\bm x,t)$. Thus, the reconstruction uses only the evolution of the aggregate density together with measured entries and exits at the outer boundary; individual vehicle velocities and internal crossing records are not used. Unless otherwise stated, a $20\times20$ grid and a $5$-s update interval are used. A spatially uniform compatibility correction is applied to the estimated forcing before the Neumann solve so that the discrete mass-balance condition is satisfied.
\begin{figure}[!t]
	\centering
	\includegraphics[width=0.5\textwidth]{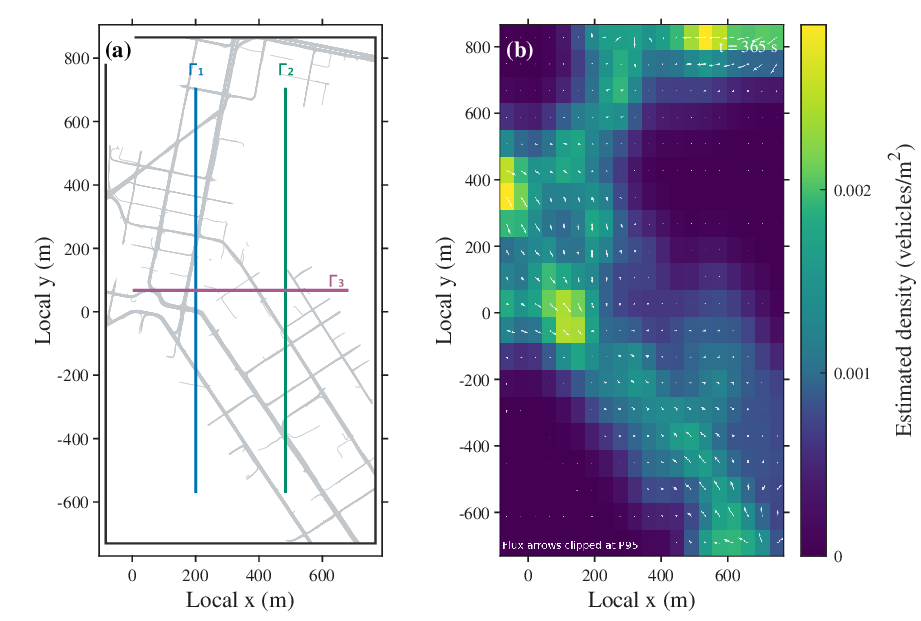}
	\vspace*{-8mm}
	\caption{Real-trace reconstruction using the pNEUMA mobility dataset. (a) Observation domain and three fixed internal interfaces. (b) Aggregate vehicle density and reconstructed information flux at a representative time.}
	\label{fig:pneuma_reconstruction}
	\vspace*{-3mm}
\end{figure}
Fig.~\ref{fig:pneuma_reconstruction}(a) shows the observation region and the three predefined internal interfaces $\Gamma_1$, $\Gamma_2$, and $\Gamma_3$. Panel~(b) gives a representative density snapshot together with the reconstructed information flux. The interfaces are fixed before any performance metric is calculated.
To evaluate the reconstructed flux, the original high-rate trajectories are used only to count vehicles crossing these interfaces. For an interface $\Gamma$, the signed crossing rate is
\begin{equation}
J_{\Gamma}^{\mathrm{trace}}(t)=\frac{N_{\Gamma}^{+}(t)-N_{\Gamma}^{-}(t)}{\Delta t_{\mathrm{upd}}}.
\end{equation}
where $N_{\Gamma}^{+}(t)$ and $N_{\Gamma}^{-}(t)$ are the numbers of crossings in the two opposite directions. The corresponding F-STSG crossing rate is obtained by integrating the reconstructed information flux across $\Gamma$ according to Proposition~\ref{Prop_Flux_Integral}.
\begin{figure}[!t]
	\centering
	\includegraphics[width=0.5\textwidth]{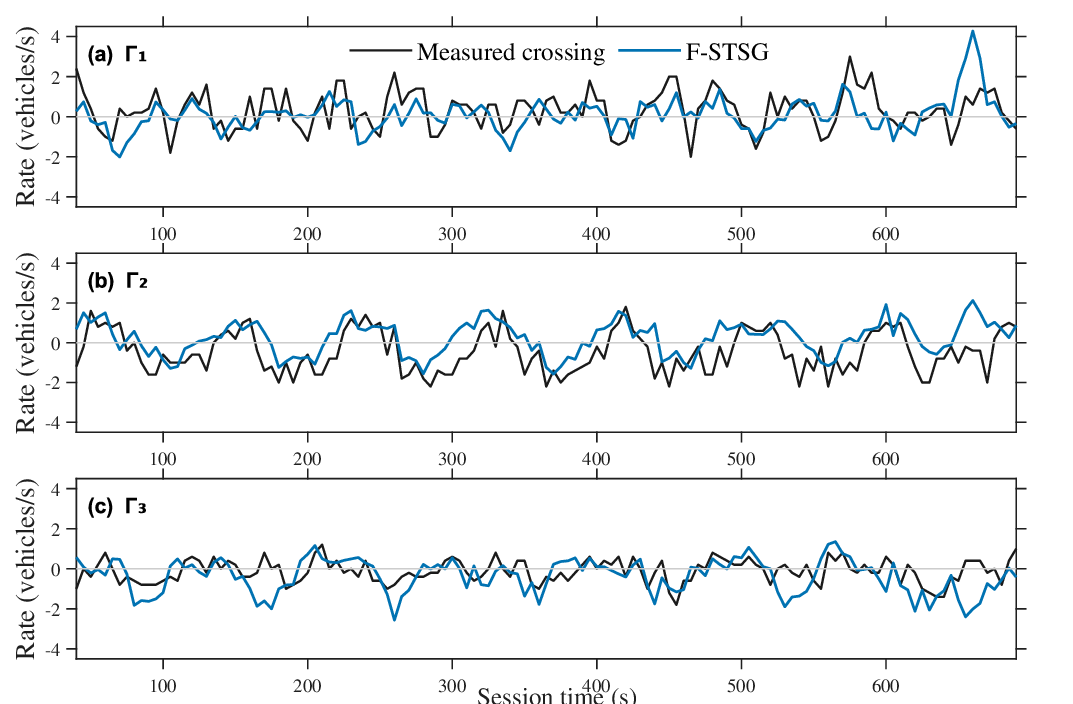}
	\vspace*{-8mm}
	\caption{Real-trace validation of reconstructed internal crossing rates on the pNEUMA dataset. Panels (a)--(c) compare the trajectory-counted and reconstructed signed crossing rates at $\Gamma_1$, $\Gamma_2$, and $\Gamma_3$, respectively.}
	\label{fig:pneuma_crossing}
	\vspace*{-3mm}
\end{figure}
Fig.~\ref{fig:pneuma_crossing} shows that the reconstructed rates capture part of the temporal variation observed in the trajectory counts, while noticeable errors remain at some time instants. As a simple reference, we also consider a centroid-translation baseline in which the whole density field moves with its centroid velocity. This baseline retains only the global drift and ignores spatially varying transport within the observation region.

\begin{table}[!t]
	\centering
	\small
	\caption{Average real-trace results over three pNEUMA sessions and three internal interfaces. MAE and RMSE are in vehicles/s.}
	\label{tab:pneuma_validation}
	\vspace*{-3mm}
	\begin{threeparttable}
		\setlength{\tabcolsep}{4pt}
		\begin{tabular}{c|cccc}
			\hline
			Method & Pearson & Spearman & MAE & RMSE\\
			\hline
			F-STSG reconstruction & 0.511 & 0.515 & 0.766 & 0.965\\
			Centroid baseline & -0.028 & -0.040 & 0.877 & 1.076\\
			\hline
		\end{tabular}
		\begin{tablenotes}
			\footnotesize
			\item F-STSG: fluid-spatiotemporal stochastic geometry; MAE: mean absolute error; RMSE: root-mean-square error.
		\end{tablenotes}
	\end{threeparttable}
	\vspace*{-5mm}
\end{table}

Table~\ref{tab:pneuma_validation} reports the averages over the three sessions and three interfaces. The F-STSG reconstruction gives Pearson and Spearman correlations of $0.511$ and $0.515$, whereas the centroid baseline gives $-0.028$ and $-0.040$. The MAE decreases from $0.877$ to $0.766$ vehicles/s, and the RMSE decreases from $1.076$ to $0.965$ vehicles/s. Averaged over the three sessions and internal interfaces, F-STSG achieves higher temporal correlation and lower MAE and RMSE than the centroid-translation baseline. Additional tests with different grid resolutions and update intervals show that very fine temporal sampling can become less reliable when too few vehicles are observed in each update. For the default $20\times20$ grid, the inverse-reconstruction stage takes $6.733$~ms per update on average, with a 95th-percentile latency of $8.715$~ms on an Intel(R) Core(TM) Ultra 7 155U CPU. These timings exclude telemetry transfer, raw-data preprocessing, field estimation, and downstream control. The experiment therefore evaluates the recovery of macroscopic transport from aggregate density evolution; it does not assess microscopic vehicle velocities or deployment-level system performance.
\subsection{Validation and Temporal Interpretation of Kinematic Metrics}\label{sec:sim_mobility_metrics}

We use a matched periodic example to verify the scalar kinematic identities and illustrate their temporal interpretation. The velocity divergence computed from $\nabla\cdot\bm v^*$ agrees with the material-derivative relation in \eqref{eqPro3}, while the centroid velocity obtained from the generated trajectory agrees with the intensity-weighted mean velocity in \eqref{eqVoC}.

\begin{figure}[!t]
	\centering
	\includegraphics[width=0.9\columnwidth]{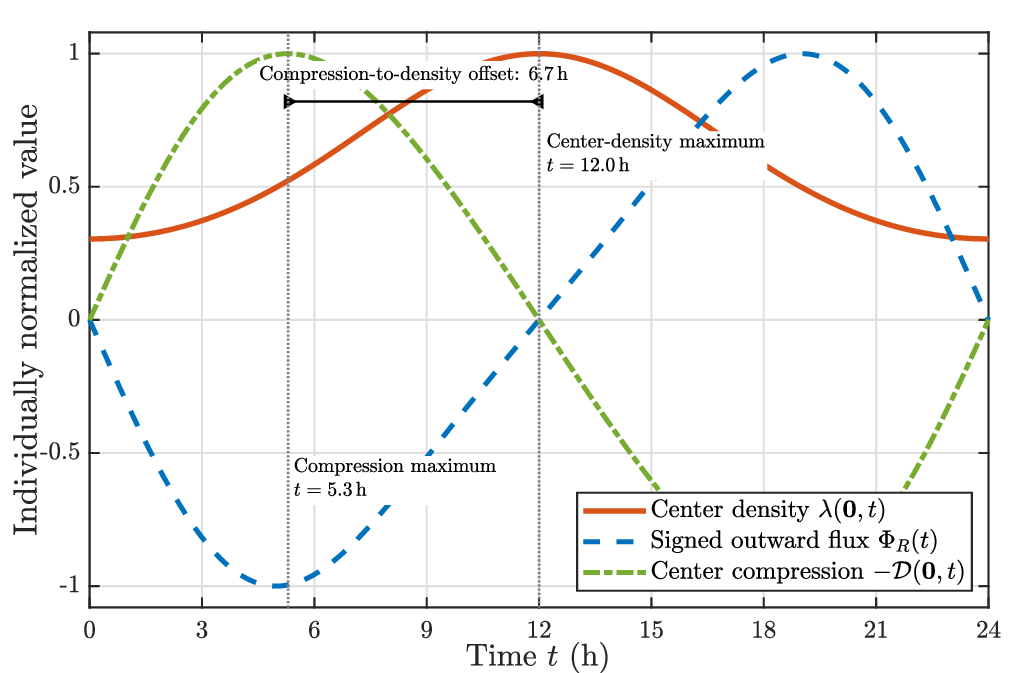}
	\vspace*{-4mm}
\caption{Normalized center density, signed outward boundary flux, and center compression indicator in the selected periodic example. Peak compression and peak inward transport precede the center-density maximum, and both indicators cross zero at the density peak.}
	\label{fig:handover}
	\vspace*{-2mm}
\end{figure}

Figure~\ref{fig:handover} illustrates the temporal relation among the kinematic quantities for the selected periodic radial example. Let $J_{\partial B_{r_0}}(t)=\int_{\partial B_{r_0}}\bm J\cdot\bm n\,\mathrm dS$ denote the signed outward flux across the boundary of a disk of radius $r_0$, so that $J_{\partial B_{r_0}}<0$ represents net inward transport and $J_{\partial B_{r_0}}>0$ represents outward transport. Radial symmetry gives $\bm v^*(\bm 0,t)=\bm 0$, and \eqref{eqPro3} therefore relates the center compression directly to the local density growth rate. Peak compression and peak inward transport precede the center-density maximum at $t=12$~h, and both indicators cross zero at the density peak. The exact lead times are specific to this periodic example.
	
	
	\begin{figure*}[b]
		\vspace*{-4mm}
		\centering
		\includegraphics[width=0.9\textwidth]{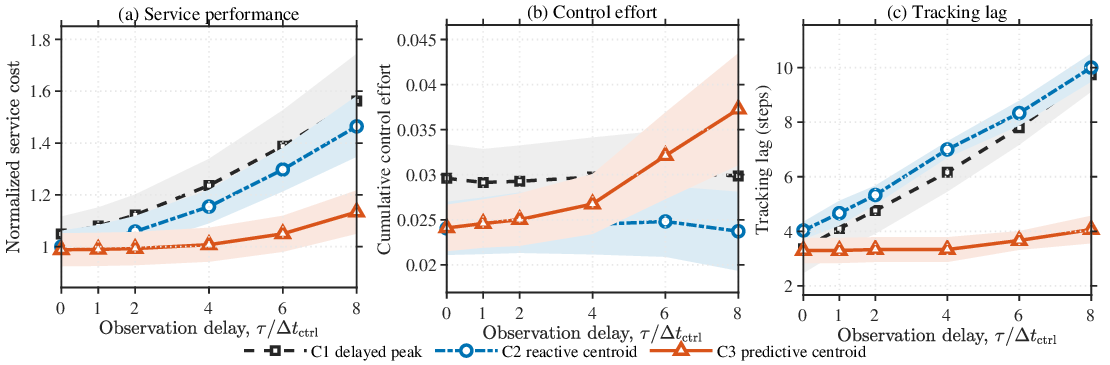}
		\vspace*{-5mm}
		\caption{Sensitivity to observation delay: (a) normalized cumulative service distance, (b) cumulative control effort, and (c) tracking lag. Markers indicate the mean, and shaded bands represent pointwise normal-approximation $95\%$ confidence intervals across $30$ scenario-and-seed-matched runs at each observation delay. Panel~(a) shows normalized group means, whereas the gains reported in the text are paired run-level averages. C4 is omitted because it uses future information and serves only as a clairvoyant reference.}
		\label{fig:control_delay_sensitivity}
	\end{figure*}
\subsection{Constrained Centroid-Reference UAV Control}\label{sec:applications}

\subsubsection{Centroid Reference and Controller}\label{sec:app_centroid}
Unlike our prior UAV-logistics application, which used information flux for workload prediction and proactive infrastructure provisioning \cite{DongDigitalTides2026}, the demonstration here uses the intensity-weighted centroid as a causal control reference.

Let
	$p(\bm x,t)=\lambda(\bm x,t)/N(t)$
	be the normalized demand density, and let $\bm C(t)$ denote its intensity-weighted centroid. Assuming that $p(\bm x,t)$ has a finite second moment, the instantaneous mean squared service distance from a mobile resource at $\bm P$ satisfies
	\begin{equation}
		\begin{aligned}
			Q(\bm P,t)
			&\triangleq
			\int_{\mathbb R^d}
			\|\bm x-\bm P\|^2
			p(\bm x,t)\,\mathrm d\bm x\\
			&=
			Q(\bm C(t),t)
			+
			\|\bm P-\bm C(t)\|^2.
		\end{aligned}
		\label{eq:uav_centroid_decomposition}
\end{equation}
The second equality follows from
	$\int_{\mathbb R^d}(\bm x-\bm C(t))p(\bm x,t)\,\mathrm d\bm x=\bm 0$.
	Hence, $\bm C(t)$ uniquely minimizes this instantaneous unconstrained squared-distance objective and is used below as a control reference.

At control update $k$, let $\tau\geq0$ denote the observation delay and assume that it is an integer multiple of the control interval $\Delta t_{\mathrm{ctrl}}$. Define	$d_{\tau}\triangleq\tau/\Delta t_{\mathrm{ctrl}}\in\mathbb N_0$
	and $t_k=k\Delta t_{\mathrm{ctrl}}$. The quantity $\widehat{\bm C}_{k+h\mid k-d_{\tau}}$ denotes the centroid predicted for update $k+h$ using only observations available through update $k-d_{\tau}$. Over a horizon of $H$ steps, let $\bm P_h$ and $\bm u_h$ denote the predicted UAV position and velocity command, respectively. Let $\mathcal F_k\subset\mathbb R^d$ be the admissible position set used at update $k$, including the flight-region bounds and the local obstacle exclusion. The predicted positions and velocity commands are obtained from
	\begin{align}
		\!\min_{\{\bm P_h\}_{h=0}^{H},
			\{\bm u_h\}_{h=0}^{H-1}}\quad \!\!\!\!&
		\sum_{h=1}^{H}
		\left\|
		\bm P_h-
		\widehat{\bm C}_{k+h\mid k-d_{\tau}}
		\right\|^2
		+
		\eta
		\sum_{h=0}^{H-1}
		\|\bm u_h\|^2
		\label{eq:uav_constrained_controller}\\
		\mathrm{s.t.}\quad&
		\bm P_0=\bm P(t_k),\nonumber\\
		&
		\bm P_{h+1}
		=
		\bm P_h+
		\Delta t_{\mathrm{ctrl}}\bm u_h,
		\, h=0,\ldots,H \!-\!1,\nonumber\\
		&
		\|\bm u_h\|
		\leq v_{\max},
		\quad h=0,\ldots,H-1,\nonumber\\
		&
		\|\bm u_h \!- \! \bm u_{h-1} \!\|
		\!\leq\! a_{\max}\Delta t_{\mathrm{ctrl}},
		 h\! =\!0,\ldots,H \!-\!1 
          ,\nonumber\\
		&
		\bm P_h\in\mathcal F_k,
		\quad h=1,\ldots,H.
		\nonumber
\end{align}
Here, $\bm u_{-1}$ is the velocity applied immediately before update $k$, $\eta\geq0$ balances the tracking and control-effort terms, and $v_{\max}$ and $a_{\max}$ are the speed and acceleration limits. Only the first command $\bm u_0$ is applied, after which the problem is solved again at the next update. The numerical quadratic program uses conservative polygonal approximations of the speed and acceleration constraints and a local linear approximation of the obstacle boundary. The optimization in \eqref{eq:uav_constrained_controller} is a planning-level kinematic reference generator: $\bm u_h$ is a commanded velocity rather than a rotor-level control input. The constraints on $v_{\max}$ and $a_{\max}$ limit the planned motion but do not model command-response dynamics, execution-layer speed or acceleration saturation, or aerodynamic disturbances. These effects are examined separately below.

\subsubsection{Evaluation Setup and Results}

The closed-loop evaluation uses three demand-evolution scenarios representing hotspot translation, turning, and abrupt migration. At each update, a fixed total of
	$N_{\mathrm{obs}}=50000$
	counts is sampled from the normalized demand field on a $25\times25$ spatial grid. The empirical centroid is the count-weighted mean of the grid coordinates, and the empirical density peak is the center of the grid cell with the largest count.

We compare delayed density-peak control (C1), delayed reactive-centroid control (C2), causal predictive-centroid control (C3), and a clairvoyant constrained reference (C4). C1 repeats the latest available density peak over the prediction horizon, C2 repeats the latest available centroid, and C3 extrapolates the centroid by a least-squares linear fit using up to the latest five available estimates. C4 uses the future true centroid and is included only as a noncausal reference.

All controllers use the same initial condition, motion constraints, objective weight, and solver. The main setting uses $H=5$, $d_{\tau}=4$, $v_{\max}=0.0500$, $a_{\max}=0.0180$, and $\eta=0.0350$ in normalized simulation units.

Let $N_{\mathrm{ctrl}}$ denote the number of control updates in one run. For run $r$ and controller $c$, let $p_r(\bm x,t)$ denote the true demand density used for evaluation, and let $\bm P_{r,k}^{c}$ and $\bm u_{r,k}^{c}$ denote the realized UAV position and applied velocity command at update $k$, respectively. For controller	$c\in\{\mathrm{C1},\mathrm{C2},\mathrm{C3}\}$	and run $r$, define the cumulative service-distance objective and control effort as
	\begin{equation}
		\begin{aligned}
			\mathcal L_{\mathrm{svc},r}^{c}
			&=
			\Delta t_{\mathrm{ctrl}}
			\sum_{k=1}^{N_{\mathrm{ctrl}}}
			\int_{\mathbb R^d}
			\left\|
			\bm x-\bm P_{r,k}^{c}
			\right\|^2
			p_r(\bm x,t_k)\,
			\mathrm d\bm x,\\
			\mathcal L_{\mathrm{ctrl},r}^{c}
			&=
			\Delta t_{\mathrm{ctrl}}
			\sum_{k=0}^{N_{\mathrm{ctrl}}-1}
			\left\|
			\bm u_{r,k}^{c}
			\right\|^2.
		\end{aligned}
		\label{eq:uav_evaluation_metrics}
\end{equation}
The first quantity is evaluated by the corresponding spatial-grid sum in the simulation. The true demand field $p_r$ is used only to evaluate all controllers against the same noise-free service objective and is not supplied to C1--C3 during control. The second quantity measures the accumulated squared velocity commands.

For Fig.~\ref{fig:control_delay_sensitivity}, all service-distance values are normalized by the mean C2 value at zero delay. The percentages reported in the text are instead obtained from the paired run-level gain. For $N_{\mathrm{pair}}$ scenario-and-seed-matched runs, the gain of C3 relative to C2 is
	\begin{equation}
	\overline{\mathcal G}_{3,2}
		=
		\frac{1}{N_{\mathrm{pair}}}
		\sum_{r=1}^{N_{\mathrm{pair}}}
		\frac{
			\mathcal L_{\mathrm{svc},r}^{\mathrm{C2}}
			-
			\mathcal L_{\mathrm{svc},r}^{\mathrm{C3}}
		}{
			\mathcal L_{\mathrm{svc},r}^{\mathrm{C2}}
		}.
		\label{eq:paired_service_gain}
\end{equation}
The $95\%$ confidence intervals use the normal approximation across the corresponding run-level quantities. For \eqref{eq:paired_service_gain}, the interval is computed from the paired gains rather than from the difference between group means.

For controller $c$ in run $r$, the tracking lag is the nonnegative temporal shift $\ell\in\{0,\ldots,L_{\max}\}$ that minimizes the RMSE between the realized UAV positions and the delayed true-centroid trajectory, where $L_{\max}=\min(12,N_{\mathrm{ctrl}}-3)$. Ties are resolved using the smallest shift.

The main comparison uses $30$ observation-count seeds in each of the three scenarios, giving $90$ paired runs. C3 achieves a $12.80\%$ mean paired service-distance gain over C2 ($95\%$ confidence interval: $11.52\%$--$14.08\%$) and a $16.47\%$ gain over C1. The C3-versus-C2 win/tie/loss counts are $90/0/0$. The mean tracking lag is $3.33$ steps for C3 and $7.00$ steps for C2. This improvement is accompanied by a $9.59\%$ increase in pooled control effort. Under nominal direct execution, no solver failure or speed, acceleration, obstacle, or flight-region violation is observed. Execution nonidealities are evaluated separately below.

The delay sweep uses $10$ seeds in each scenario, giving $30$ paired runs at every delay. As shown in Fig.~\ref{fig:control_delay_sensitivity}, the mean paired C3-over-C2 gain increases from $1.29\%$ at zero delay to $21.94\%$ at an eight-step delay, together with higher control effort. In the abrupt-migration scenario, the gain decreases to $5.81\%$, indicating that linear centroid extrapolation is less effective when the demand trajectory changes abruptly. The study covers aggregate-count centroid estimation, causal centroid prediction, and constrained UAV tracking, but not the complete PDE-reconstruction-to-control chain.

\subsubsection{Execution-Layer Stress Test}
We further examine the planning-to-execution gap by passing the commanded velocity $\bm u_k$ through a surrogate UAV execution layer. Let $\bm w_k$ denote the executed velocity. The command-response, acceleration, and position updates are modeled as
\begin{equation}
	\begin{aligned}
			\bm a_k^{\mathrm{cmd}}&=\frac{\bm u_k-\bm w_k}{\tau_a},\\
			\bm a_k^{\mathrm{exe}}&=\operatorname{sat}_{a_{\mathrm{phys}}}\!\left(\bm a_k^{\mathrm{cmd}}+\bm d_k\right),\\
			\bm w_{k+1}&=\operatorname{sat}_{v_{\mathrm{phys}}}\!\left(\bm w_k+\Delta t_{\mathrm{ctrl}}\bm a_k^{\mathrm{exe}}\right),\\
			\bm P_{k+1}&=\bm P_k+\Delta t_{\mathrm{ctrl}}\bm w_{k+1},
	\end{aligned}
	\label{eq:uav_execution_model}
\end{equation}
where $\tau_a$ represents the finite command-response time, $v_{\mathrm{phys}}$ and $a_{\mathrm{phys}}$ impose effective speed and acceleration saturation, and $\bm d_k$ is a bounded temporally correlated wind-acceleration disturbance. For a vector $\bm z$, let $\operatorname{sat}_{c}(\bm z)\triangleq\bm z\min\{1,c/\|\bm z\|\}$ for $\bm z\neq\bm 0$, with $\operatorname{sat}_{c}(\bm 0)=\bm 0$, so that the Euclidean norm is limited to $c$. This surrogate model is intended to test execution mismatch rather than reproduce full six-degree-of-freedom quadrotor dynamics.

The same three demand scenarios, 30 seeds, controller parameters, and initial conditions are retained. Five execution conditions are considered: nominal direct execution; a response lag with $\tau_a=2\Delta t_{\mathrm{ctrl}}=2$~s; tightened physical limits $v_{\mathrm{phys}}=0.8v_{\max}$ and $a_{\mathrm{phys}}=0.8a_{\max}$; bounded correlated wind with root-mean-square magnitude $0.1a_{\max}$, maximum magnitude $0.2a_{\max}$, and correlation time $3\Delta t_{\mathrm{ctrl}}$; and a combined condition containing all three effects. Paired controllers use the same wind realization. A tracking-failure update is recorded when the centroid-tracking distance exceeds the threshold $0.16$.

\begin{table}[!t]
	\centering
	\caption{Execution-layer stress test for C3. Gain is the paired service-distance gain over C2. Speed saturation, acceleration saturation, and tracking failure are reported as percentages of control updates; obstacle violations are reported as update counts.}
	\label{tab:uav_execution_stress}
	\vspace*{-2mm}
	\scriptsize
	\setlength{\tabcolsep}{2.2pt}
	\begin{tabular}{lcccccc}
		\hline
		Case & RMSE & Gain & Speed sat. & Accel. sat. & Track fail. & Obstacle\\
		\hline
		Nominal & 0.0834 & 12.80\% & 0 & 0 & 6.11\% & 0\\
		Lag & 0.0936 & 12.96\% & 0 & 0 & 7.19\% & 0\\
		Saturation & 0.0897 & 12.46\% & 8.48\% & 20.61\% & 6.67\% & 0\\
		Wind & 0.0842 & 12.81\% & 2.83\% & 7.37\% & 6.11\% & 0\\
		Combined & 0.0958 & 12.68\% & 8.52\% & 5.07\% & 7.74\% & 24\\
		\hline
	\end{tabular}
	\vspace*{-2mm}
\end{table}

Table~\ref{tab:uav_execution_stress} shows that execution nonidealities produce measurable tracking degradation. A response time of $\tau_a=2\Delta t_{\mathrm{ctrl}}$ increases the C3 tracking RMSE by $12.18\%$ relative to nominal execution, while no saturation or constraint violation occurs in the lag-only case. Under tightened physical limits, the speed and acceleration bounds are activated in $8.48\%$ and $20.61\%$ of the C3 updates, respectively. Wind alone increases the C3 tracking RMSE by $0.89\%$ under the tested disturbance level. Under the combined condition, the paired C3-over-C2 service gain remains $12.68\%$, compared with $12.80\%$ nominally, but 24 obstacle-violation updates occur in 13 of the 90 C3 runs. In the abrupt-migration scenario, the tracking-failure update rate is $14.56\%$, and the C3-over-C2 gain is $4.44\%$. No flight-region violation or C3 solver failure is observed. These results show that execution mismatch can produce effective speed or acceleration saturation, tracking degradation, and local feasibility violations under the tested surrogate model; they do not constitute six-degree-of-freedom flight-dynamics or hardware validation.

\section{Conclusion}\label{sec:conclusion}

This paper established the kinematic foundation of F-STSG for non-stationary wireless networks. Based on population conservation, macroscopic transport reconstruction was formulated as a weighted inverse boundary value problem. Under the stated positivity and compatibility conditions, the minimum-energy principle selects a unique admissible velocity representative, from which information flux, divergence, and intensity-weighted centroid are derived. The formulation also establishes the identifiability boundary of density-invisible circulation.

Numerical experiments examined reconstruction consistency and sensitivity, while real pNEUMA mobility traces provided empirical validation of internal transport against held-out crossing measurements. A separate UAV study evaluated centroid-based control under observation delays and execution perturbations, revealing both performance gains and practical limitations. Future work will extend the kinematic framework to wireless interference and performance analysis, investigate causal prediction of evolving demand fields, and evaluate end-to-end implementations in operational cellular and UAV systems.


	\small

\makeatletter
\def\@IEEEBIOskipN{0.4\baselineskip}
\makeatother

\begin{IEEEbiography}
	[{\includegraphics[width=1in,height=1.25in,clip,keepaspectratio]{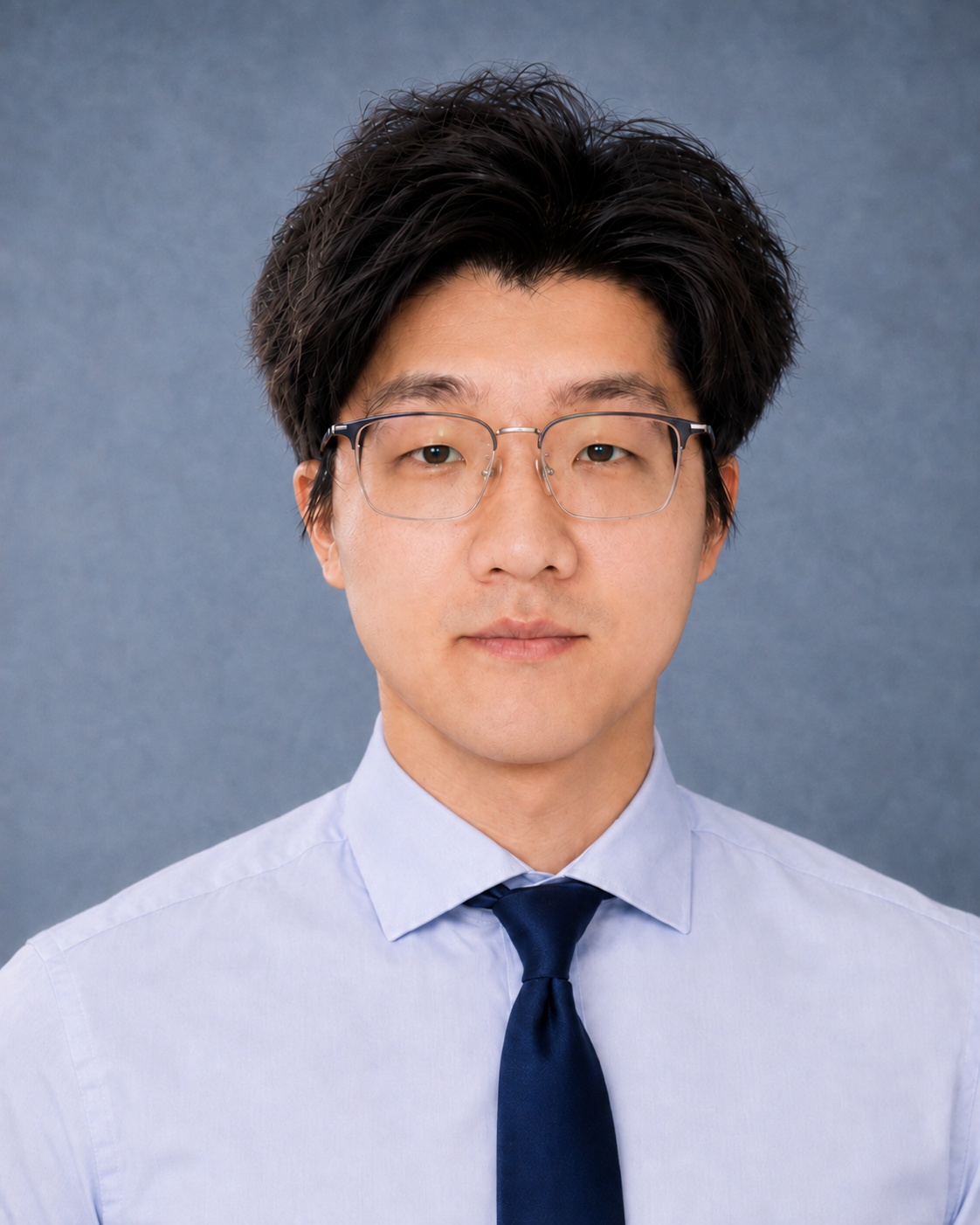}}]
	{Wen-Yu Dong}
	received the B.S. degree in electronic and information engineering from Sichuan University, Chengdu, China, in 2019, and the Ph.D. degree in information and communication engineering from the School of Information and Communication Engineering, Beijing University of Posts and Telecommunications, Beijing, China, in 2025. He is currently a Researcher with the Future Technology Research Center, China Telecom Research Institute, Beijing, China. His current research interests include space-air-ground integrated networks, fluid-spatiotemporal stochastic geometry (F-STSG), and wireless communications.
\end{IEEEbiography}

\begin{IEEEbiography}
	[{\includegraphics[width=1in,height=1.25in,clip,keepaspectratio]{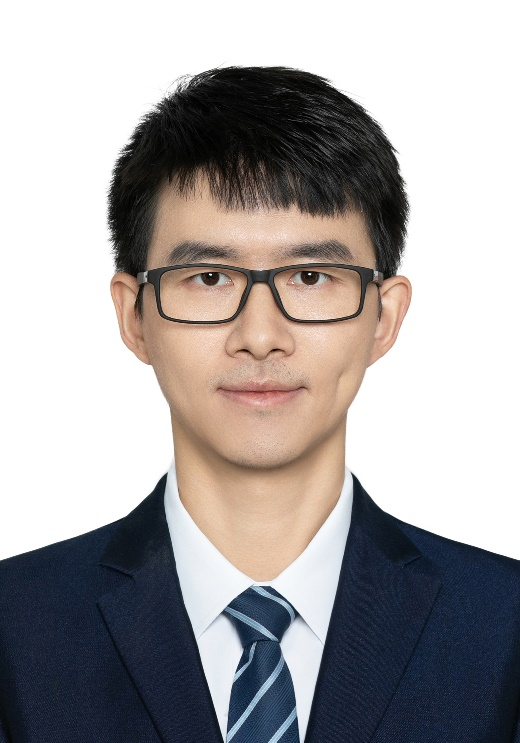}}]
	{Weiwei Jiang} (Senior Member, IEEE)
	received the B.Sc. Degree of Electronic Engineering and Ph.D. Degree of Information and Communication Engineering from the Department of Electronic Engineering, Tsinghua University, Beijing, China, in 2013 and 2018, respectively. He is currently an associate professor with the School of Information and Communication Engineering, Beijing University of Posts and Telecommunications, and Key Laboratory of Universal Wireless Communications, Ministry of Education. His current research interests include artificial intelligence, machine learning, big data, wireless communication and edge computing. He has published more than 100 academic papers in IEEE Trans and other journals, with more than 5700 citations in Google Scholar. He was included in Stanford's World's Top 2\% Scientists list in 2022--2026.
\end{IEEEbiography}

\begin{IEEEbiography}[{\includegraphics[width=1in,height=1.25in,clip,keepaspectratio]{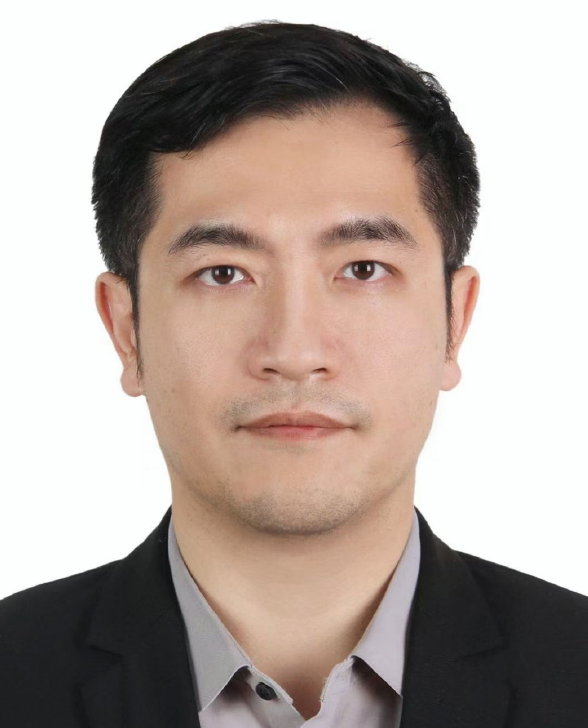}}]{Song Zhao}
	is a senior engineer with the Future Technology Center, China Telecom Research Institute, Beijing, China. He received his Ph.D. degree in Telecommunications and Information Systems from Beijing University of Posts and Telecommunications, Beijing, China. He serves as a delegate of China Telecom in the 3GPP SA2 and SA5 working groups, and as rapporteur for multiple 3GPP work items related to network automation and intelligence. His research interests include wireless channel modeling, radio resource management, proximity services, and network AI.
\end{IEEEbiography}

\begin{IEEEbiography}[{\includegraphics[width=1in,height=1.25in,clip,keepaspectratio]{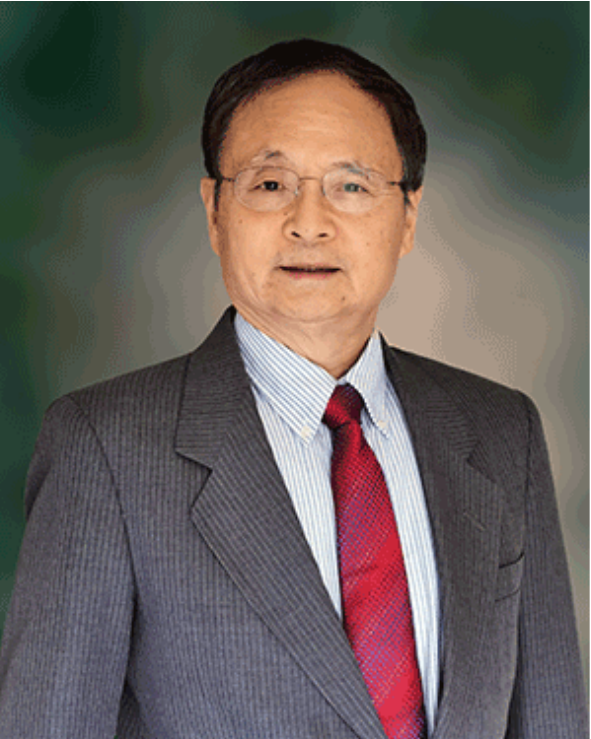}}]{Qi Bi}
	(Fellow, IEEE) is Chief Scientist of China Telecom and CTO of its Research Institute, specializing in 5G and 6G. He received his M.S. from Shanghai Jiao Tong University and Ph.D. from Pennsylvania State University. He was awarded Bell Labs Fellow in 2002, the Bell Labs President’s Gold Awards in 2000 and 2002, the Asian American Engineer of the Year in 2005, the Beijing Outstanding Contribution Award for Innovation and Entrepreneurship for Overseas Scholars in 2019, and three Patent Silver Prizes from the China National Intellectual Property Administration from 2022 to 2024.
\end{IEEEbiography}

\begin{IEEEbiography}[{\includegraphics[width=0.9in,height=1.1in,clip,keepaspectratio]{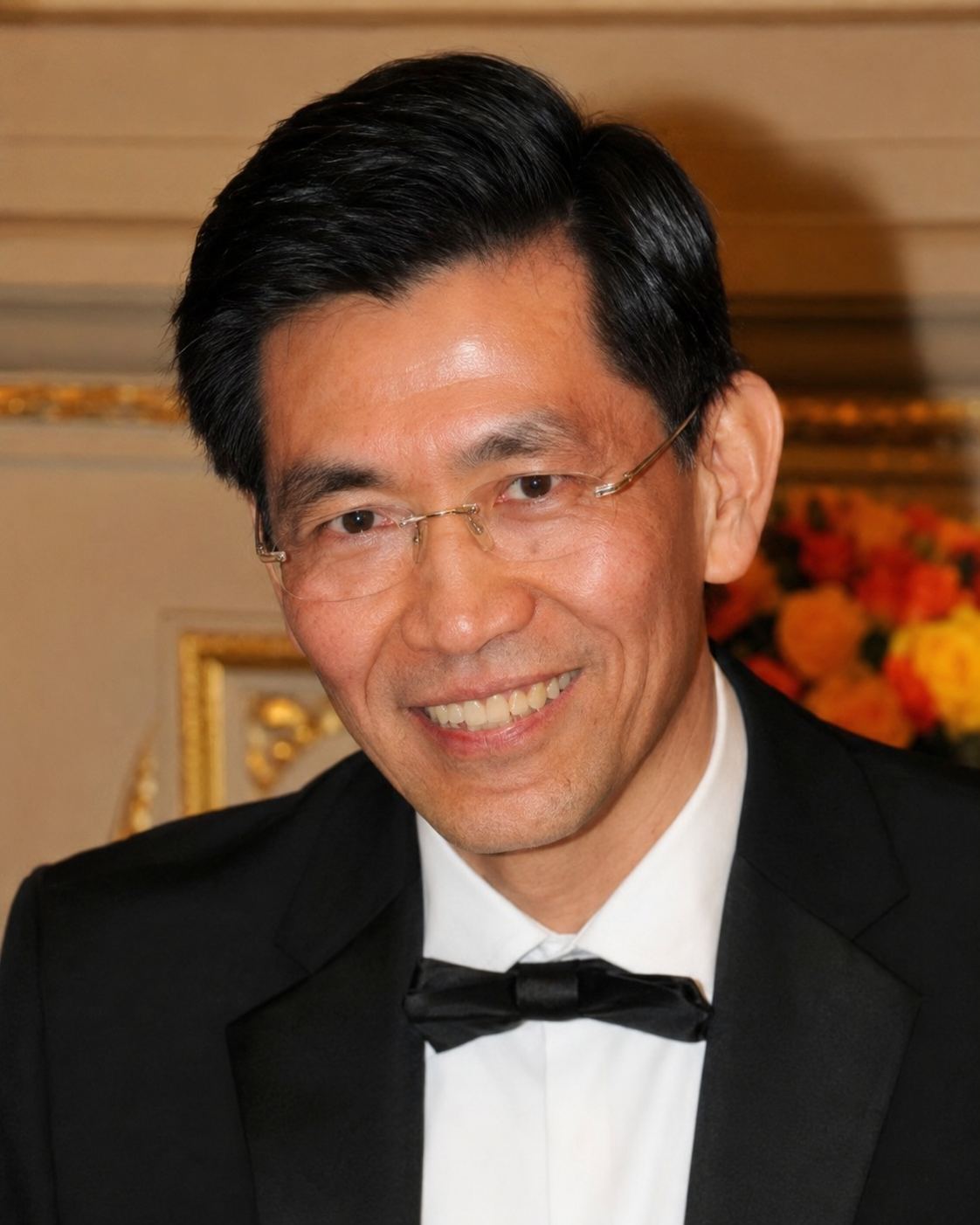}}]{SHENG CHEN} (IEEE Life Fellow) is an Emeritus Professor, University of Southampton, Southampton, U.K.
	
	He received the B.Eng. degree in control engineering from the East China Petroleum Institute, Dongying, China, in January 1982, the Ph.D. degree in control engineering from City University, London, in September 1986, and the higher doctoral (D.Sc.) degree from the University of Southampton, Southampton, U.K., in August 2005. From 1986 to 1999, he held research and academic appointments with the University of Sheffield, the University of Edinburgh, and the University of Portsmouth, all in U.K. In September 1999, he joined the School of Electronics and Computer Science, University of Southampton, and retired in June 2026. Since then, he is a University Visitor associated with the School of Electronics and Computer Science, University of Southampton. In 2023, Dr. Chen was appointed  a Distinguished Professor at Ocean University of China.
	
	Professor Chen has established an over 40-year distinguished academic career in the fields of adaptive signal processing, wireless communications, neural networks and machine learning. He has authored/coauthored over 750 research papers. He has 23,000+ Web of Science citations with h-index 67, and 44,000+ Google Scholar citations with h-index 88. Professor Chen is IEEE ComSoc Signal Processing and Computing for Communications (SPCC) 2025 Technical Recognition Award recipient.
	
	Professor Chen was elected as a Fellow of the United Kingdom Royal Academy of Engineering in 2014. He is a fellow of Asia-Pacific Artificial Intelligence Association, a fellow of Industry Academy of the International Artificial Intelligence Industry Alliance, a fellow of IET, and one of the original 200 ISI Highly Cited Researchers in engineering (March 2004).
\end{IEEEbiography}
	
\end{document}